\documentclass[USenglish]{article}

\usepackage[utf8]{inputenc}%(only for the pdftex engine)
\usepackage{babel}
\usepackage{amsmath,amsthm,amssymb,amsfonts}
\usepackage{mathrsfs}
\usepackage{bbm,bm}
\usepackage{color}
\usepackage{latexsym}
\usepackage{graphicx}
\usepackage{enumerate}
\usepackage{subcaption}
\usepackage{caption}
\usepackage{multirow}

\usepackage{algorithm2e}
\usepackage{makecell}
\usepackage{booktabs}
\usepackage{hyperref}

\newcommand{\argmax}{\operatornamewithlimits{argmax}}
\newcommand{\argmin}{\operatornamewithlimits{argmin}}
\renewcommand{\Re}{{\mathbb{R}}}
\renewcommand{\Pr}{{\mathbb{P}}}
\newcommand{\Ex}{{\mathbb{E}}}

\newcommand{\tr}{{\textup{trace}}}

\newcommand{\PM}{{\mathcal{P}}}
\newcommand{\DSM}{{\mathcal{D}}}
\newcommand{\AM}{{\mathcal{A}}}

\newcommand{\ErdosRenyi}{{Erd{\H o}s-R{\'e}nyi}}
\newcommand{\CER}{{\mathrm{CorrER}}}

\newtheorem{theorem}{Theorem}
\newtheorem{proposition}[theorem]{Proposition}
\newtheorem{lemma}[theorem]{Lemma}
\newtheorem{corollary}[theorem]{Corollary}
\theoremstyle{definition}
\newtheorem{definition}[theorem]{Definition}

\theoremstyle{remark}
\newtheorem{remark}[theorem]{Remark}

\begin{document}

 % \articletype{...}
 
  %\author[1]{Fei Fang}

  %\runningauthor{...}
 % \affil[1]{Department of Mathematics and Statistics, Boston University,  Boston, MA, 02215, E-mail: feifang@bu.edu}
 
 \title{Tractable Graph Matching via Soft Seeding}
 \author{Fei Fang \thanks{Department of Decision Sciences, Duke University, Durham, NC (\href{mailto:ff33@duke.edu}{ff33@duke.edu}).},
 Daniel L. Sussman\thanks{Department of Mathematics and Statistics, Boston University, Boston, MA (\href{mailto:sussman@bu.edu}{sussman@bu.edu}).}, and
 Vince Lyzinski\thanks{Department of Mathematics, University of Massachusetts, Amherst, MA (\href{mailto:vlyzinski@umass.edu}{vlyzinski@umass.edu}).}}

 % \runningtitle{Graph Matching, Continuous Optimization}
  %\subtitle{...}
  %\classification[PACS]{...}
  %\communicated{...}
  %\dedication{...}
  %\received{...}
  %\accepted{...}
  %\journalname{Statistics \& Risk Modeling}
  %\journalyear{...}
  %\journalvolume{..}
  %\journalissue{..}
  % \startpage{1}
  % \aop
  %\DOI{...}

\maketitle

\abstract{
The graph matching problem aims to discover a latent correspondence between the vertex sets of two observed graphs.
This problem has proven to be quite challenging, with few satisfying methods that are computationally tractable and widely applicable.
The FAQ algorithm \cite{FAQ} has proven to have good performance on benchmark problems and works with a indefinite relaxation of the problem.
Due to the indefinite relaxation, FAQ is not guaranteed to find the global maximum.
However, if prior information is known about the true correspondence, this can be leveraged to initialize the algorithm near the truth.
We show that given certain properties of this initial matrix, with high probability the FAQ algorithm will converge in two steps to the truth under a flexible model for pairs of random graphs.
Importantly, this result implies that there will be no local optima near the global optima, providing a method to assess performance.
 }

\section{Introduction}

Tools to analyze multiple graphs together and to study the relationship between graphs are emerging rapidly \cite{Kivela2014-ij,Durante2018-pl,Levin2017-uk}.
While numerous methods can be employed which essentially treat the graphs as unlabeled and employ graph invariants, when the correspondence between vertices is known, this enables more detailed comparison and the ability to study how individual nodes behave across different networks \cite{Vogelstein2015-xh}.
If there is some true but unobserved correspondence between a pair of graphs, accurately estimating this correspondence can allow for the use of numerous methods which would be precluded without such a matching.

Graph matching is the problem of discovering the latent correspondence between a pair of graphs by trying to find a correspondence which ensures the two graphs are close to each other in terms of edit distance \cite{ConteReview,Emmert-Streib2016-st}.
Let $\AM{}_n$ and $\PM{}_n$ denote the sets of $n\times n$ adjacency and permutation matrices, respectively.
Formally, for $A,B\in \AM_n$, the graph matching problem seeks an element of
\begin{equation}\label{eq:gm_intro}
  \argmin_{P \in \PM_n} \|A - PBP^T\|_F^2,
\end{equation}
which is a permutation which minimizes the number of entries $i,j$ where $A_{ij}\neq (PBP^T)_{ij}$.
Each minimizing permutation encodes the estimated bijection between the node sets of the two graphs.

Two questions immediately arise:
When does solving the graph matching problem find the true correspondence between the vertex sets? 
Can the above graph matching problem be solved with a computationally tractable algorithm?
For the first questions, a number of authors have considered the setting where $A,B$ are distributed according some joint distribution where there is positive correlation between corresponding entries of $A$ and $B$ \cite{FAP,lyzinski2016information}.
We continue in this vein as described in Section~\ref{sec:model}.

With regards to algorithmic tractability, the graph matching problem is equivalent to the quadratic assignment problem \cite{finke1987quadratic} which is known to be NP hard and has eluded even polynomial approximations.
Hence, rather than directly trying to find a solution to this problem, we instead ask, how much prior information is necessary so that tractable algorithms will yield the desired correspondence?

The computational challenges of the graph matching problem can be alleviated if some prior information about the correspondence is known.
This prior information comes in the form of seed nodes, or anchors, for which the latent correspondence is known.
For a pair of social networks, these might be users with the same user name and location;
for a pair gene networks, these might be genes with common DNA sequence;
and for a pair of brain networks these may be two regions with the same structure, function, and relative location.
Provided a sufficient number of seeds are known, a problem that is initially computationally intractable can become solvable using relatively fast procedures in polynomial time \cite{lyzinski2014seeded,yartseva2013performance,kazemi2015growing}.

However, frequently the prior information may itself be noisy, with errors and uncertainty. 
In the social network context, users with the same user name likely correspond to the same person but this is not guaranteed.
Similarly, if only geolocations are available then we may want to account for the fact that geolocations may not be static across networks.
In general, given similarities---ranging from being very likely to being very unlikely to be a match across the graphs---between pairs of nodes across networks, we can consider using this noisy and uncertain information in the graph matching procedure.

In the remainder of this manuscript, we explore the idea of using prior information about the true correspondence to construct an initialization matrix (Section~\ref{sec:seed}) for a relaxation of the graph matching problem (Section~\ref{sec:faq}).
We consider this procedure within the correlated \ErdosRenyi{} model (Section~\ref{sec:model}), and we show that, provided the initialization contains enough information about the true correspondence, the algorithm will converge rapidly to the true correspondence between node sets (Section~\ref{sec:theory}).
In Section~\ref{sec:theory}, we outline the proofs and discuss the assumptions of the results, but we leave the details of the proofs to the Appendix.
Our theoretical results are further explored in various simulations (Section~\ref{sec:sim}).

\section{Background}

For the remainder of the paper we will use the following notation.
For $n\in \mathbb{N}_+$, let $[n]= \{1,2,\dotsc, n\}$.
Let $\AM_{n} = \{0,1\}^{n\times n}_{sym}$ be the set of adjacency matrices on graphs with $n$ vertices.
Let $\PM_{n}$ denote the set of $n\times n$ permutation matrices, $\{P: P\in\{0,1\}^{n\times n}, PP^T = I\}$, and let $\DSM_{n}$ denote the set of doubly stochastic matrices, $\{D\in \Re^{n\times n}: D_{ij} \geq 0, \sum_{k=1}^n D_{kj} = \sum_{k=1}^n D_{ik} = 1, \forall i, j \in [n] \}$.
Let $J_n$ be the $n\times n$ matrix of all ones and $I_n$ be the $n\times n$ identity matrix.
Often we will omit the subscripts, in which case we assume all matrices are $n\times n$.
Let $\|\cdot \|_F$ denote the Frobenius norm, $\otimes$ the Kronecker product, and $\circ$ the Hadamard product.

\subsection{Fast Approximate QAP Algorithm}
\label{sec:faq}

Since permutation matrices are unitary and $\|\cdot\|_F$ is unitarily invariant, Eq.~\eqref{eq:gm_intro} is equivalent to 
\begin{equation}
  \argmax_{P\in \PM_n} \tr(APBP^T).
\end{equation}
In \cite{FAQ}, the authors propose using a relaxation of this objective function from $\PM_n$ to its convex hull, the set of doubly stochastic matrices $\DSM{}_n$.
The relaxed problem is 
\begin{equation}\label{eq:gm_relax}
  \argmax_{D \in \DSM_n} \tr(ADBD^T).
\end{equation}
Other authors have proposed alternative relaxations \cite{alex,Egozi2013-jh}, however these relaxations will frequently not be guaranteed to have their global maximum correspond to the true latent correspondence \cite{Lyzinski2016-kp}.
On the other hand, for the random graph models which we consider, the optimum of this relaxation will with high probability coincide with the optimum of the original problem (see \cite{Lyzinski2016-kp,lyzinski2017graph} for details).

Algorithm~\ref{alg:faq} describes the gradient ascent approach proposed in \cite{FAQ}.
The gradient of Eq.~\eqref{eq:gm_relax} at $D$, viewed as a matrix, is given by $ADB$.
The Frank-Wolfe algorithm \cite{FW} proceeds by first finding an ascent direction which remains within $\DSM_n$ by solving the linear assignment problem (LAP).
The LAP can be solved in polynomial time using a variety of algorithms including the Hungarian algorithm~\cite{hungarian}, and the Jonker-Volgenant algorithm~\cite{Jonker1987-bs}.
% More recently, there have been developments on approximate algorithms to solve the LAP {\color{red} cite}.

In more detail, step~\ref{al:grad} approximates the objective function with a linear function and optimizes that linear function over the set of permutations, the vertices of the polytope $\DSM$.
In Step~\ref{al:line}, a line search is performed along the segment between $D_k$ and $P_k$, which is easily solved as the objective function is a quadratic function along any line.
The final step of the algorithm projects the doubly stochastic solution onto the set of permutation matrices, which also can be formulated as an instantiation of the linear assignment problem.

\begin{algorithm}[h!]
\caption{Fast approximate quadratic assign program (FAQ) algorithm \cite{FAQ} for graph matching.}\label{alg:faq}
\KwData{$A,B\in\AM$, $D_0\in \DSM$, $k=0$}
\While{not converged}{
\nl $P_k\leftarrow \argmax_{P\in \mathcal{P}} \tr( A D_k B P)$\; \label{al:grad}
\nl $\alpha_k\leftarrow \argmax_{\alpha\in [0,1]} \tr( A D_\alpha  B D_\alpha)$,\\\hspace*{24pt} where $D_{\alpha_k}= \alpha D_k + (1-\alpha) P_k$\; \label{al:line}
\nl  $D_{k+1}\leftarrow D_{\alpha_k}$ and $k\leftarrow k+1$\; 
}
\nl Project  $D_{k}$ onto $\PM{}$ yielding $P^*$\; \label{al:project}
\nl Return $P^*$
\end{algorithm}

Convergence may be assessed by either assessing the change in the objective function Eq.~\eqref{eq:gm_relax} or the change in the matrix $D_k$ according to an appropriate norm.
As discussed below, frequently $D_k$ will end at a permutation which is a local maximum.
Hence, convergence is often easy to assess, and the final projection step is often unnecessary.

\subsection{Model}
\label{sec:model}

% {\color[RGB]{255,100,0} loss function is $\ell(D, P) = n-\tr(DP^T)$}

The correlated heterogeneous \ErdosRenyi{} model provides a joint distribution for a pair of random graphs wherein each graph is marginally distributed with independent edges, but the adjacencies between the same set of nodes across the two graphs are correlated.

\begin{definition}[Correlated heterogeneous \ErdosRenyi{}]
Suppose $\Lambda \in [0,1]^{n\times n}$ and $R\in [0,1]^{n\times n}$.
A pair of adjacency matrices $(A,B)\sim \CER(\Lambda, R)$ if $A, B \in \AM{}_n$, and 
\begin{enumerate}
  \item for each $1 \leq u < v \leq n$, $A_{uv}$ are independent with $A_{uv}\sim \mathrm{Bernoulli}(\Lambda_{uv})$.
  \item for each $1\leq u < v \leq n$, $B_{uv}$ are independent with $B_{uv}\sim \mathrm{Bernoulli}(\Lambda_{uv})$.
  \item Additionally, $B_{u,v}$ and $A_{u',v'}$ are independent unless $\{ u, v \} = \{ u', v' \}$, in which case $\mathrm{corr}(A_{uv}, B_{uv}) = R_{uv}$.
\end{enumerate}
Finally, for a given $\Lambda, R$, let $Q$ be the matrix of edgewise covariances with entries $Q_{ij} = R_{ij} \Lambda_{ij} (1-\Lambda_{ij})$.
\end{definition}

% independence within graphs
% dependence across graphs
% state as $(A_{uv}, B_{uv}) \stackrel{ind}{\sim} CorrBern(\Lambda_{ij}, R_{ij})$ for all $u,v$ and $(A_{vu}, B_{vu}) = (A_{uv}, B_{uv})$

Under the correlated \ErdosRenyi{} model, the true correspondence between the node sets is given by the identity matrix, and for the remainder of the paper we will assume this, without loss of generality.
Following this convention, the loss function that we consider is $$\ell(D) = n-\tr(D)$$ for any doubly stochastic matrix $D$.

\begin{remark}
Our definition of the correlated \ErdosRenyi{} model has been proposed elsewhere, including \cite{FAP,lyzinski2017graph} among others.
This model also clearly extends models such as the homogeneous \ErdosRenyi{} model where $\Lambda_{ij} = p$ for all $i,j$.
Another popularly studied model is one wherein there is an underlying graph $G$ and $A$ and $B$ are adjacency matrices corresponding to randomly sampled subgraphs of $G$ \cite{kazemi2015growing}.
Provided the distribution for $G$ and the sampling schemes both maintain independence across node-pairs, the correlated \ErdosRenyi{} model will include the resulting distribution for $A$ and $B$. 

The case where $A$ and $B$ are not-identically distributed is not included in the definition provided above.
While we do not consider this case explicitly in the manuscript, in the discussion we discuss how ideas from \cite{lyzinski2017graph} can be employed. 
\end{remark}

\section{Seeds}\label{sec:seed}

In this section, we will explore how seeded vertex information can be incorporated into the FAQ algorithmic framework.

\subsection{Hard Seeds}\label{sec:hard_seed}

In the hard seeding case, a part of the true correspondence between the node sets is known.
From the perspective of the FAQ algorithm, this corresponds to treating certain rows and columns of the permutation matrix as fixed in the graph matching optimization formulation.
Assuming that the seeds (letting $s$ denote the number of seeds) indicate that the $s\times s$ principal submatrix is set to the identity, the optimization procedure then becomes 
\begin{equation}
  \argmax_{D\in \DSM{}_n}  2 \tr(A_{12}D B_{12}^T) + \tr(A_{22} D B_{22} D) \label{eq:sgm}
\end{equation}
where $A_{12}, B_{12}$ denote the submatrices corresponding to the seed rows and non-seed columns and $A_{22},B_{22}$ denote the submatrices corresponding to the non-seed rows and columns.

Many authors \cite{FAP,kazemi2015growing,feizi2016spectral} have shown that using seeds will substantially improve performance, even when the number of seeds is not substantially large.
For example, by only optimizing the first term in Eq.~\eqref{eq:sgm}, the optimum can be computed in polynomial-time by solving a single linear assignment problem, and in some cases only $\theta(\log n)$ seeds are needed to guarantee exact recovery of the true correspondence \cite{FAP} in the \ErdosRenyi{} setting.
Via an alternate strategy leveraging typicality matchings and methods from point--to--point communication analysis, a polynomial-time algorithm that guarantees exact recovery has been demonstrated in the $s=\theta\left(\frac{\log n}{I(X_1,X_2)}\right)$ (where $I(X_1,X_2)$ is the mutual information across the edge distributions) regime \cite{shirani2017seeded}.
Percolation algorithms leveraging very few seeds are proposed in \cite{kazemi2015growing}, and are shown to recover almost all matches in the \ErdosRenyi{} setting using only $\Theta(1)$ seeds.

\subsection{Soft Seeds}\label{sec:soft_seed}

One issue with hard seeds is that seed errors will persist and perhaps introduce further errors in the remaining matching.
Furthermore, hard seeds impose that the prior information about the true correspondence must be one-to-one between nodes.
To alleviate this, we can treat seeded vertices as ``soft seeds'' by setting the $s\times s$ principal submatrix of the initialization of Eq. (\ref{eq:gm_relax}) to the identity, and proceeding with the unconstrained optimization.  

Moreover, soft seeding allows us to utilize prior correspondence information that is not exact knowledge of the one--to--one correspondence.
In general, the prior information may come in the form of similarity scores between nodes across the two node sets.
Our proposal for soft seeding uses this information to construct a suitable doubly stochastic matrix $D^0$ from which we will initialize the FAQ algorithm (see Section \ref{sec:construct_ds} for detail).

As the FAQ algorithm employs an indefinite relaxation, the algorithm can be very sensitive to the initial point.
If the similarity scores yield a doubly stochastic initialization which is close enough to the global optimum, then even though the problem is indefinite, the gradient ascent procedure will ideally yield the global optimum.
Our theoretical goals for the remainder of this manuscript are to study the set of doubly stochastic initializations that will guarantee FAQ terminates at the correct permutation.

\subsubsection{Constructing doubly stochastic initializations}\label{sec:construct_ds}

Construction of the matrix $D^0$ will depend on the specific prior information available.
The simplest form of prior information will consist of a one-to-one map between subsets of $[n]$ which can be represented as a subset $S \subset [n]^2$.
Given $S$, a standard way to construct a matrix $D^0$ will be to set $D^0_{ij} = 1$ for all $(i,j)\in S$, set $D^0_{ij} = 0$ for all $(i,j)$ where $(i,j')\in S$ for some $j'$ or $(i',j)\in S$ for $i'$.
The remaining entries are set to $1/(n-|S|)$, corresponding to the barycenter of the doubly stochastic matrices of size $n-|S|$.

A more general form of prior information provides subsets of vertices in one graph which correspond to sets of vertices in the other graph.
This can be represented as a pair of partitions, $\eta_1, \eta_2,\dotsc, \eta_s$ and $\zeta_1,\zeta_2, \dotsc,\zeta_s$ of $[n]$, with $|\eta_k|=|\zeta_k|$ for each $k\in [s]$.
The soft seeding then provides that nodes in $\eta_k$ in the first graph correspond to nodes in $\zeta_k$ in the second graph.
The matrix $D^0$ can then be constructed as $D^0_{ij} = 1/|\eta_k|$ for every $i\in \eta_k$ and $j\in \zeta_k$ for each $k$.
All other entries are set to zero.
Each submatrix with rows and columns in a given element of the partition will  correspond to the barycenter of the doubly stochastic matrices of size $|\eta_k|$.
If $\eta_1=\zeta_1, \eta_2=\zeta_2, \dotsc, \eta_s = \zeta_s$, then $\tr(D^0) = s$.
Note that, as a special case, we could consider that each of $\eta_1,\eta_2,\dotsc, \eta_{s-1}$ are singletons and $\eta_s$ and $\zeta_s$ contain the remaining vertices.
This corresponds to the one-to-one soft seeding described above.

Prior information may also come in the form of node attributes or features which are believed to be relatively similar between corresponding nodes in the two graphs.
Suppose $X_1,\dotsc, X_n\in \Re^d$ and $Y_1,\dotsc, Y_n\in \Re^d$ denote features for each node.
Given a similarity function or kernel $\kappa:\Re^d \times \Re^d \mapsto \Re^+$, we can construct a similarity matrix $S\in \Re_+^{n \times n}$ where $S_{ij} = \kappa(X_i, X_j)$.
The matrix $S$ is not guaranteed to be doubly stochastic, but we can construct a doubly stochastic matrix by setting $D^0 = \argmin_{D\in \DSM_n} \|D - S\|_F$.
Solving this problem is rather straightforward and relatively simple algorithm to do so can be found in \cite{Wang2010-og}, which also suggests other methods for learning a doubly stochastic matrix.
A fast, commonly used (less principled) approach to find a doubly stochastic matrix corresponding to $S$ is to iteratively rescale the rows and columns by their inverse row-sums and column-sums respectively.
This is called the Sinkhorn-Knopp algorithm \cite{Sinkhorn1967-eo} and is known to converge and be unique under certain conditions on $S$.

Random doubly stochastic matrices may also be useful for further exploring the set of local optima \cite{Sussman2018-lb}.
A basic example is to take $D^0$ as a random permutation matrix.
Another possibility is to sample a random similarity matrix $S$ and construct $D^0$ as described above.
In general, we can take any convex combination of doubly stochastic matrices to get another doubly stochastic matrix which would allow us to combine these various approaches.
Some tools to construct and sample doubly stochastic matrices are provided in the R package \texttt{iGraphMatch} \cite{Qiao2018-xy}, which also implements the FAQ algorithm.

\section{Theoretical guarantees in FAQ}\label{sec:theory}

The parameters for the correlated \ErdosRenyi{} model, $\Lambda$ and $R$, and the initial matrix $D^0$ will impact whether the FAQ procedure will yield the correct correspondence.
If $A,B\sim \CER{}(\Lambda, R)$, then the most important aspect of $D^0$ will be its trace, as this provides a quantification for the amount of correct information contained in the seeding. 
When $D^0 = I$, $\tr(D^0)=n$ represents fully correct seed information about all correspondences. 
When $D^0 = \frac{1}{n} J$, $\tr(D^0)=1$ which, in this case, indicates essentially no information is provided by the seeding, as each node is equally likely to correspond to each other node.

In the case of homogeneous \ErdosRenyi{} marginals with constant correlation, $A, B \sim \CER{}(p J, r J)$, almost sure polynomial-time performance guarantees via FAQ only require that the $\tr(D^0)$ is sufficiently large.
In particular, we have the following.
\begin{theorem}\label{thm:hom}
Suppose $A,B\sim \CER{}(p J, r J)$. 
Let $\delta \in (0, 1 / 2)$.
Set $\ell = \frac{2\sqrt{n^{1+2\delta}}(3p(1-p) + 2r)}{(rp(1-p))^2}$.

Let $C=rp(1-p)$ and $\epsilon=3p(1-p)+2r$.
Let $\mathcal{P}_{n,\ell}$ be the set of permutation matrices with trace at least $\ell$.
With probability at least 
\[
    1  - 2 \exp\left\{- \Theta \left( \frac{C^4}{\epsilon^2} n^{2-\delta} \log n  \right) \right\},
\]
the FAQ procedure will converge in at most two steps to the identity matrix when started at any doubly stochastic matrix which is a convex combination of permutation matrices in $\mathcal{P}_{n,\ell}$.
\end{theorem}

% For the random graph parameters, both lower bounds on the covariance between entries and upper bounds on the total variance for the relevant random variables will be important. 
% To use the concentration inequality, Proposition~\ref{prop:kim}, with $A,B\sim \CER{}(\Lambda, R)$, we need that $A,B$ can be realized as a function of independent Bernoulli random variables.
% Proposition~\ref{prop:bibern} provides a construction for bivariate Bernoulli random variables as a function of three independent Bernoulli random variables.
% The sum of the variances of the three independent variables is bounded over all $i,j$ by $\max_{ij} 3\Lambda_{ij}(1-\Lambda_{ij} + 2 R_{ij}$, which we will denote as $\epsilon$.
% We will also use that $\mathrm{Cov}(A_{ij}, B_{ij}) = R_{ij} \Lambda_{ij} (1-\Lambda_{ij})$ which we will denote as $S_{ij}$.

Rather than directly prove this theorem, we use a more general result for heterogeneous distributions.
To this end, note that for any permutation matrices $P,Q$,
\begin{align}
 \Ex[ \tr(APBQ^T) ] = & \tr(\Lambda P \Lambda Q^T) \label{eq:trex1} \\
 & + \sum_{i\neq j} S_{ij} (P_{ii} Q_{jj} + P_{ji} Q_{ij}) \label{eq:trex2},
\end{align}
where $S_{ij}=\mathrm{Cov}(A_{ij}, B_{ij}) = R_{ij} \Lambda_{ij} (1-\Lambda_{ij})$.
In the homogeneous setting, $\Lambda = p J$ and $R = r J$, the first term does not depend on $P,Q$, and the second term is equal to $rp(1-p) (\tr(P)\tr(Q) + \tr(PQ - 2 P \circ Q)$, 
where ``$\circ$'' denotes the Hadamard product which here quantifies the number of common fixed points between $P$ and $Q$.
In the heterogeneous case, the term $\tr(\Lambda P \Lambda Q^T)$ will depend on the structure of $\Lambda$ and can possibly decrease as $\tr(Q)$ increases.
On the other hand, the second term will again vary with $P$ and $Q$ approximately as the product of their traces.
Control of these two terms is key in the heterogeneous setting, as we see in the following theorem. 

\begin{theorem}\label{thm:het}
Suppose $A,B\sim \CER{}(\Lambda, R)$.
Let $\epsilon = \max_{i,j} 3\Lambda_{ij}(1-\Lambda_{ij} + 2R_{ij}$.
Let $\delta, C >0$ satisfy $\delta < 1 / 2$, and $C < \min_{i,j} R_{ij} \Lambda_{ij}(1-\Lambda_{ij})$.

Set $\ell = 2\sqrt{n^{1+2\delta}}\epsilon/C^2$ and $m=C^2 n^{1-\delta}\log n/\epsilon$.
%  and 
% \[
% \ell = 2\sqrt{n^{1+2\delta}}\epsilon/C^2
% \text{ and }
% m = C^2 n^{1-\delta}\log (n) / \epsilon.
% \]

Suppose that 
\begin{equation}
  \Ex[\tr(APB(Q-Q')] \geq C \tr(P)\tr(Q - Q') \label{eq:ex_ass}
\end{equation}
 for all permutation matrices $P, Q, Q'$ with $\tr(P) \geq \ell$ and $\tr(Q) > n - m$ and $\tr(Q)>\tr(Q')$.

Let $\mathcal{P}_{n,\ell}$ be the set of permutation matrices with trace at least $\ell$.
With probability at least 
\[
    1  - 2 \exp\left\{- \Theta \left(  \frac{C^4}{\epsilon^2}n^{2-\delta} \log n  \right) \right\},
\]
the FAQ procedure will converge in at most two steps to the identity matrix when started at any doubly stochastic matrix which is a convex combination of permutation matrices in $\mathcal{P}_{n,\ell}$.
\end{theorem}

Note that in the homogeneous setting, the condition in Eq.~\eqref{eq:ex_ass} below is automatically satisfied for $C \lessapprox r p (1 - p)$ (see the discussion below for a precise form).
In general, Eq~\eqref{eq:ex_ass} is difficult to verify, but we suspect (and have found empirically) it should hold quite widely.
Note that it may not necessarily hold that $\tr(Q)>\tr(Q')$ implies $\tr(\Lambda P \Lambda (Q-Q'))>0$.
However, this implies that there is latent structure in $\Lambda$ which for which $Q'$ provides a better alignment than $Q$.
In the present setting, we may not be able to overcome this, as this represents the case where there are potential symmetries in the model which occur with non-trivial probability.

Note that, as argued in \cite{lyzinski2017graph}, by centering the matrices $A$ and $B$ according to an estimate of $\Lambda$, the term in Eqs. (\ref{eq:trex1}) can be made sufficiently small so that the expectation is dominated by Eq.~\eqref{eq:trex2}.
Additionally, it can frequently hold that Eq.~\eqref{eq:trex1} will increase with the trace of $P$ and $Q$, when the structure of $\Lambda$ itself is informative about the alignment.

Provided $\Ex[\tr(\Lambda P \Lambda (Q-Q'))]$ is non-negative or a sufficiently small negative value for the relevant permutations, Eq.~\eqref{eq:ex_ass} will hold provided we can lower bound 
\begin{equation}
  \sum_{i=1}^n \sum_{j\neq i} S_{ij} (P_{ii} (Q_{jj}-Q'_{jj}) + P_{ji} (Q_{ij}-Q'_{ij})) \label{eq:trex2_diff}
\end{equation}
based on Eq.~\eqref{eq:trex2}.
Returning again to the homogeneous \ErdosRenyi{} case (or in the heterogeneous case where $S$ is a constant matrix), a lower bound on this will hold provided that for all $P$ and all $i,j$ with $i\neq j$, it holds that $\sum_{k\neq i} P_{kk} - P_{ij} \geq C'\tr(P)$, for some $C'>0$.
For any $P$, in the homogeneous setting this inequality holds with $C' = 1 -\frac{1}{\tr(P)} > 1 - \frac{1}{\ell}$.
Hence, we can take $C$ to satisfy $C = \left( 1 - \frac{1}{\ell} \right) p r (1-p)$.

In the more general case, it may not be the case that a lower bound on Eq.~\ref{eq:trex2_diff} will hold for all $Q$, $Q'$.
Indeed,  $Q'$ could have slightly lower trace but align more highly  correlated parts of the graph pair than $Q$ which would lead to Eq.~\eqref{eq:trex2_diff} being negative.
To alleviate this, we note that we state Theorem~\ref{thm:het} with these relatively strong assumptions but our proof technique can apply for a smaller subset of permutation matrices $P,Q, Q'$ and our proof technique will apply to this subset.

We will now outline our proof technique.
Our proofs proceed in three main stages.
First, we prove a technical Lemma showing that $\tr(APB(Q-Q'))$ will concentrate around its expectation and hence have the correct sign with high probability.
As another stepping stone, a pointwise result indicates that, with high probability, FAQ will converge in one step to the identity for a fixed $D^0\in \DSM$ with $\tr(D^0)$ big enough.

Building on the pointwise result, we show that, with high probability, the FAQ algorithm will converge in one step for all $D^0$ with $\tr(D^0)=n-o(n)$.
The final stage of our proof is to show that, with high probability, the FAQ algorithm will take its first step to a matrix with trace $n-m=n-o(n)$ for all $D^0$ with $\tr(D^0)$ large but less than $n-m$.
This stage we divide into two propositions.
In Proposition~\ref{prop:step_dir}, we show that the search direction for the first step is towards a matrix with large trace, and in Proposition~\ref{prop:step_size}, we show that a maximal step size is used.
As the one-step result is uniform as well, we conclude the two-step result.

We remark that Theorem~\ref{thm:pw}, Corollary~\ref{cor:unif_os}, and Proposition~\ref{prop:step_dir}, all depend on the assumption in Eq.~\eqref{eq:ex_ass} only with $Q=I$.
Hence, in Eq.~\eqref{eq:trex2_diff} $P_{ii}(Q_{jj}-Q'_{jj})$ is always nonnegative, so the assumptions for these results will more easily hold.

As a final note regarding our assumptions, it is certainly possible that Eq.~\ref{eq:ex_ass} will hold only for a subset of the relevant matrices.
In this case, our proof technique as outlined would still function but the two step result would hold only for the corresponding subset of $\DSM$.

\begin{remark}
In the appendix we prove results for the case that $D^0$ is a permutation matrix with sufficiently large trace.
We do this for ease of notation and in order to employ a union over finite sets of permutations. 
While the theorems above are stated in terms of any convex combination of permutation matrices with sufficiently large trace, the theory in the appendix extends to this convex set as the theorems in the appendix provide probabilistic bounds for all permutations matrices simultaneously, and hence for all convex combinations thereof as $\tr(APBQ^T)$ is linear in $P$.
\end{remark}

% Each of the above theorems hold uniformly across $D^0$.
% A pointwise bound that has slightly weaker conditions on the trace of $D^0$ can also be proved as stated in Theorem~\ref{thm:pw}.
% Indeed, the pointwise theorem serves as a stepping stone to the uniform result.
% In the next section, we discuss important implications of the uniform result in the underlying optimization framework.

Before discussing the implications of the above theorems on the existence of bad (i.e., with a large trace) local maxima to the graph matching objective function, note the algorithmic implications of the results in the FAQ framework.
Recall that the final step of the FAQ algorithm, Algorithm~\ref{alg:faq}, is to project the final doubly stochastic matrix to the nearest permutation matrix.
The above results indicate that frequently this step will not be needed since the final doubly stochastic matrix will itself be a permutation matrix.

\subsection{Implications of the soft-seeding theory}

In this section, we consider some of the algorithmic implications of the soft seeding theory, Theorems \ref{thm:hom} and \ref{thm:het}.

\subsubsection{Local maxima}\label{sec:loc_max}
Theorems~\ref{thm:hom} and \ref{thm:het} both have strong implications on the local maxima of the relaxation in Eq.~\eqref{eq:gm_relax}.
Recall, that this problem is non-convex and in general can have many local maxima, which is the main motivation for our soft seeding approach.

In the correlated \ErdosRenyi{} setting, if the event described in Theorem~\ref{thm:het} holds then this implies that there will be no local maxima $D$ with $\tr(D)>\ell  = 2\sqrt{n^{1+\delta}}\epsilon/C^2$.
Furthermore, if $\Ex[\tr(ADBD^T)]$ scales with $\tr(D)^2$, local maximum with $\tr(D)<\ell$ will have an objective function significantly less than the objective function for the global maximum.
This observation can be used to develop heuristics to verify whether the found match is believed to be accurate or not.
When random restarts (i.e., randomized initializations $D^0$) are applied, we can explore the set of local maximum and ideally seek out the single local maximum with much larger objective function, as this will correspond to the global maximum and the true correspondence.

\subsubsection{The good and the bad of hard seeds}

If there are $s$ hard seeds, with $b<s$ of these seeds misspecified, then these label errors will cascade throughout the hard-seeded FAQ procedure.
If however, we treat all $s$ seeds as soft seeds with one--to--one correspondences, then provided $s-b>\ell  = 2\sqrt{n^{1+\delta}}\epsilon/C^2$  in the correlated \ErdosRenyi{} setting, with high probability FAQ will converge to the identity in at most two steps.
While $\sqrt{n^{1+\delta}}\gg \log(n)$ (recalling $\Theta(\log n)$ hard seeds are needed for almost sure perfect recovery in the restricted seeded FAQ formulation in \cite{lyzinski2014seeded}), soft seeding is able to leverage this additional information contained in the good seeds to mitigate the effect of the badly specified seeds.
This further suggests soft seeding as a practical alternative to hard seeding when the veracity of some of the hard seeds is in question.

% \begin{assumption}\label{A1}
% Suppose $(A,B)\sim \CER(\Lambda, R)$ for $\Lambda,R \in {[0,1]}^{n\times n}$ satisfying,
%  \begin{align}
%   \epsilon_1 &= \min\limits_{i\neq z\in[n] } \rho_{zi} \Lambda_{zi}(1-\Lambda_{zi}) > 0,\label{a:e1}\\
%   \epsilon_2&=\max\limits_{\substack{\sigma(z) \neq z\\ i\in[n]}} \left[3\Lambda_{zi}(1-\Lambda_{zi})+2\rho_{zi} \right]\text{ and,}\label{a:e2}\\
%   % \epsilon_3&=\min\limits_{\substack{\sigma(z) \neq z\\ i\in[n] }} 3(1-\rho_{zi})\Lambda_{zi}(1-\Lambda_{zi}) > 0. \label{a:e3}
% \end{align}
% Let $\epsilon = \epsilon_1/\epsilon_2$.
% \end{assumption}

% \begin{assumption}\label{A2}
% Let $D^0 \in \DSM{}$ satisfy:
%   \begin{enumerate}
%     \item For all permutation matrices $P\in \PM{}$, $\tr ( \Lambda D^0 \Lambda (I-P^T))\geq 0$.
%     \item For any $z \neq \tilde{z}$, there exist some constant $C\in(0,1)$ such that $$\sum_{i=1}^n D^0_{ii} Q_{zi} - D^0_{z\tilde{z}}Q_{z\tilde{z}}\geq C\tr D^0\epsilon_1.$$
%     \item 
%     \end{enumerate}
% \end{assumption}

\section{Simulations}\label{sec:sim}

To illustrate the utility of soft seeding and its impact on performance, we will investigate a variety of simulated settings.
The theoretical results indicate that if started at a matrix $D^0$ with sufficiently large trace, after two steps the FAQ algorithm will converge at the identity with high probability.
In the following we will investigate this as well as performance beyond two iterations.
In the \ErdosRenyi{} setting, we measure matching accuracy via the number of vertices mapped correctly to their corresponding vertices.
If the output of an algorithm is $P$, our accuracy measure is $\tr(P)/n$.
For an intermediate doubly stochastic matrix output, we will use the same measure of accuracy, $\tr(D)/n$ which can be measured at each iteration of the FAQ procedure.

We will consider three different instances of the $\CER$ model to investigate the impact of $\Lambda$ on overall performance.

\begin{description}
  \item[HOM] The Homogeneous \ErdosRenyi{} setting where $$A,B\sim  \CER(0.5 J_{300}, 0.5 J_{300}).$$
  \item[SBM] The stochastic blockmodel setting where 
  $$A,B \sim \CER \left( \ \left( 0.4 I_6 + 0.1 J_6 \right) \otimes J_{50}, 0.5 J_{300}\right),$$ so that marginally $A,B$ are from a stochastic blockmodel with 6 blocks with 50 vertices each. The probability of an edge between two vertices in the same block is $0.5$ and for two nodes in different blocks is $0.1$.
  \item[RDPG] The random dot product graph setting with $$A,B\sim \CER(XX^T, 0.5 J_{300}),$$ where $X=[X_1,\dotsc,X_{300}]^T\in \Re^{300\times d}$ and $$X_1,\dotsc,X_{300} \stackrel{iid}{\sim} \mathrm{Dirichlet}(1,1,1)_{1:2},$$ the first two coordinates of a uniform distribution on the 3-simplex. This is analogous to a degree-corrected mixed-membership stochastic block model with 2 blocks.
\end{description}

\subsection{Accuracy over Iterations}

First, we will illustrate the performance as the trace of the initial matrix $D^0$ is given by $s \in \{1, 5, 6, 8, 13, 35, 83\}$.
If $s$ divides $n$, we construct $D^0$ as $I_{s} \otimes \frac{s}{n} J_{n/s}$.
Hence, $D^0$ will be a block diagonal block, with constant blocks with entries $s/n$ of size $n/s$; this yields $\tr(D^0)=s$.
If $s$ does not divide $n$, the first $s-1$ blocks have size $\lfloor s/n \rfloor$ and the final block has size $n-(s-1) \lfloor s/n \rfloor$.
Each diagonal block of $D^0$ is itself a constant doubly stochastic matrix, the barycenter for the set of doubly stochastic matrices of the given size.
When $s=1$, $D^0$ is simply the barycenter of $\mathcal{D}_n$.

\begin{figure}
  \centering
    \begin{subfigure}[t]{.48\textwidth}
    \centering
    \includegraphics[width=\linewidth]{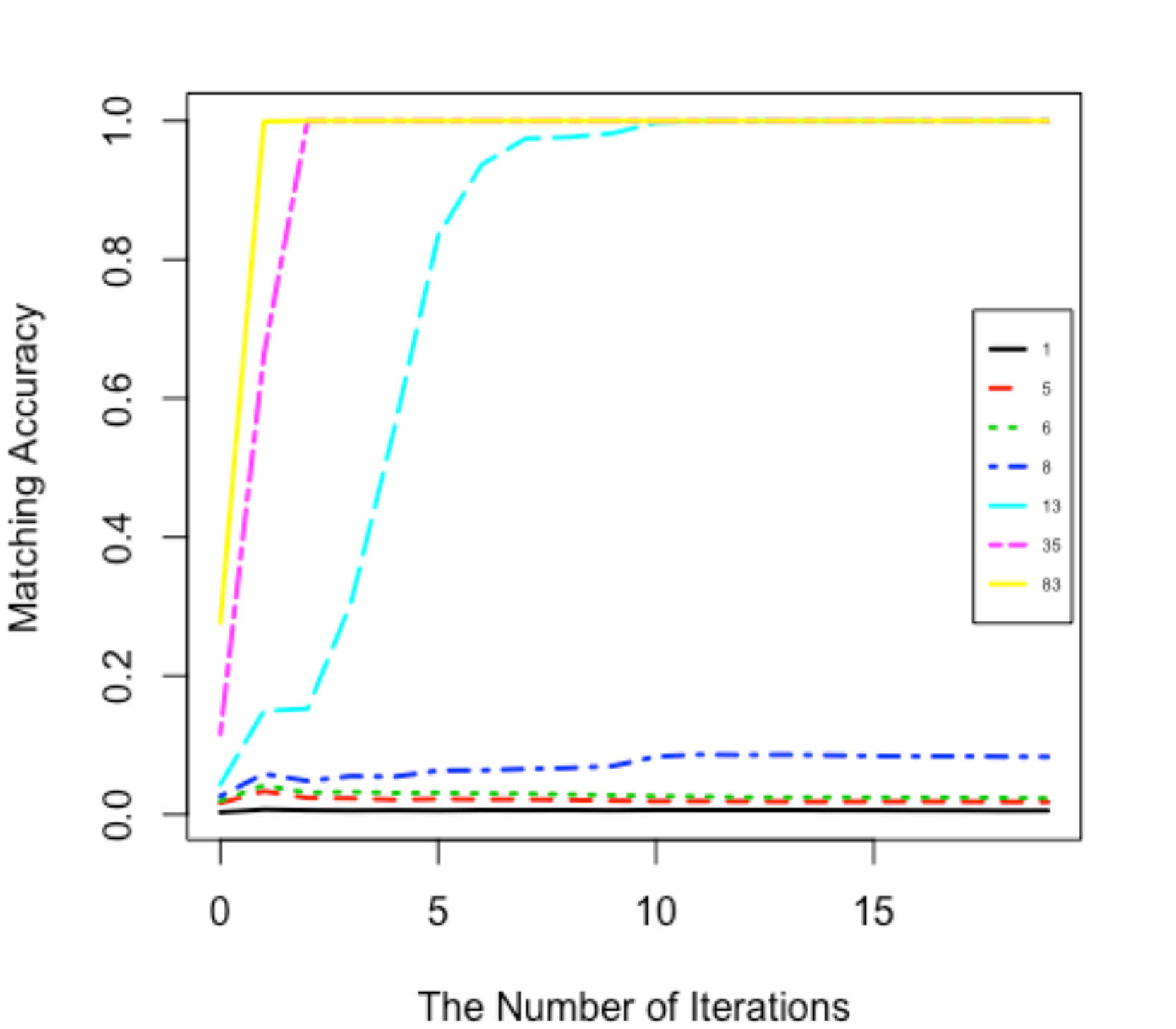}
    \caption{ HOM average accuracy}
    \label{fig:er}
  \end{subfigure}\hfill
  \begin{subfigure}[t]{.48\textwidth}
    \centering
    \includegraphics[width=\linewidth]{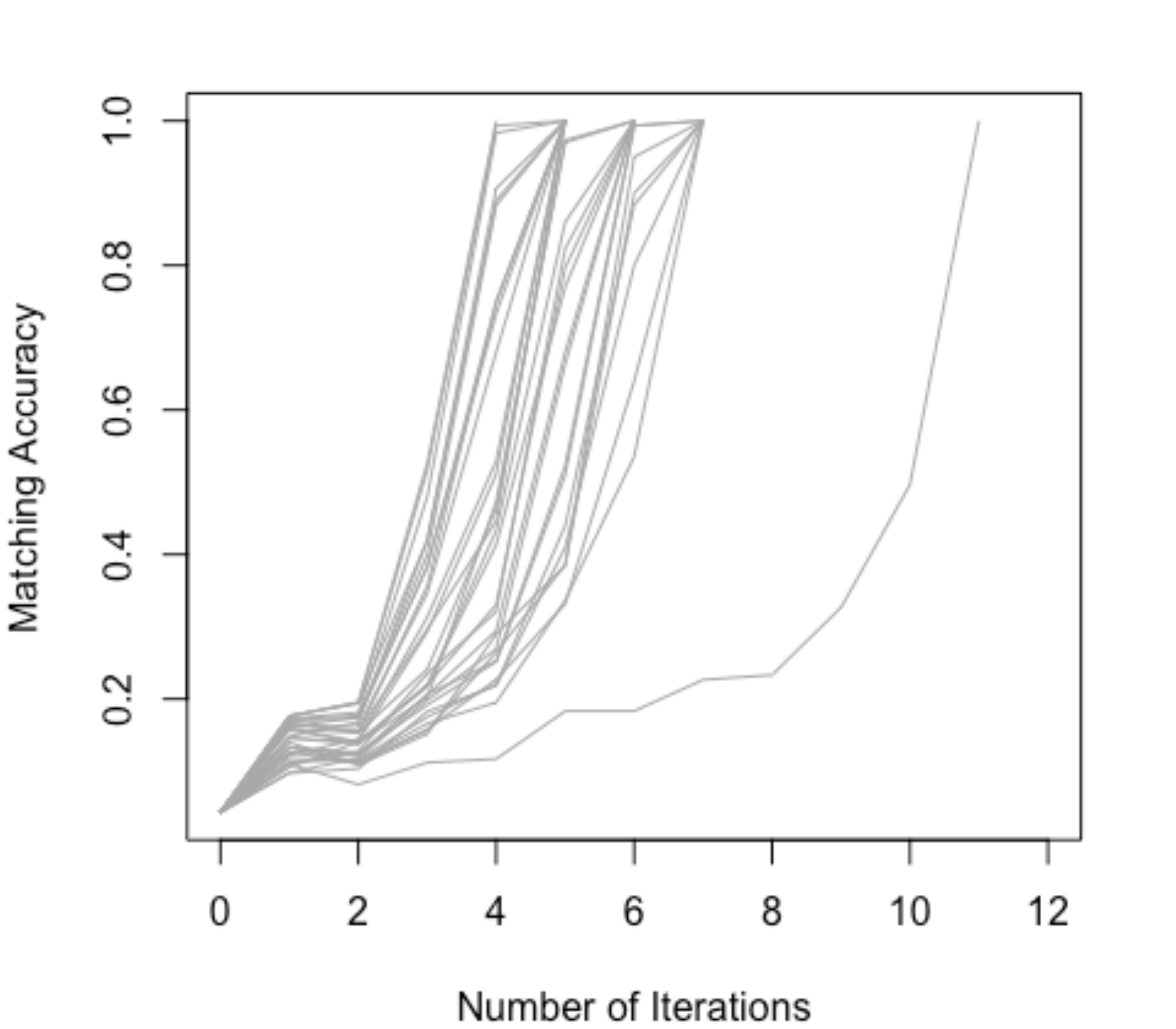}
    \caption{ HOM accuracy trajectories}
    \label{fig:er_traj}
  \end{subfigure}\\
  \begin{subfigure}[t]{.48\textwidth}
    \centering
    \includegraphics[width=\linewidth]{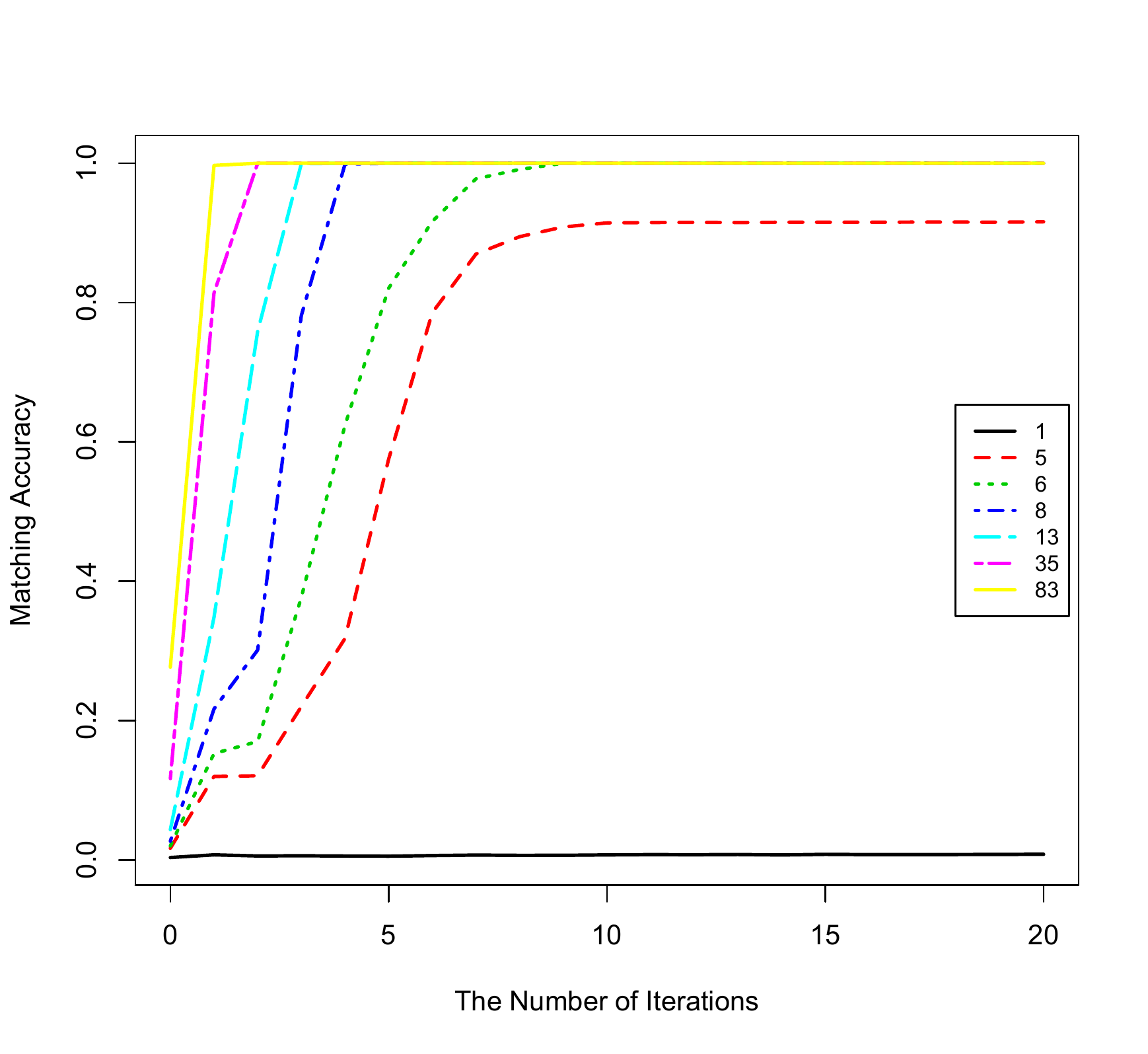}
    \caption{ SBM Average Accuracy}
    \label{fig:sbm}
  \end{subfigure}
  \begin{subfigure}[t]{.48\textwidth}
    \centering
    \includegraphics[width=\linewidth]{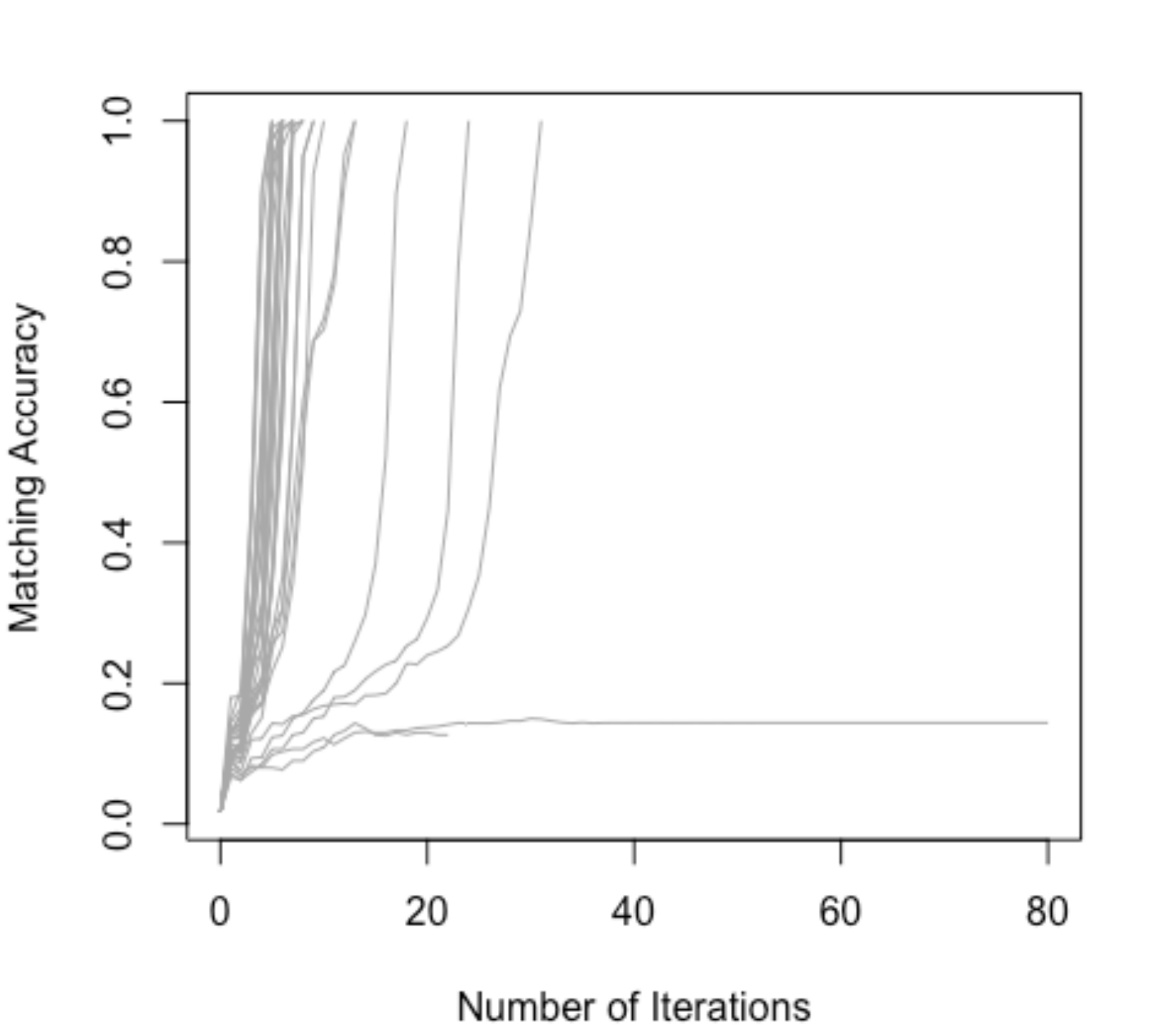}
    \caption{ SBM Accuracy trajectories}
    \label{fig:sbm_traj}
  \end{subfigure}\\
  \begin{subfigure}[t]{.48\textwidth}
    \centering
    \includegraphics[width=\linewidth]{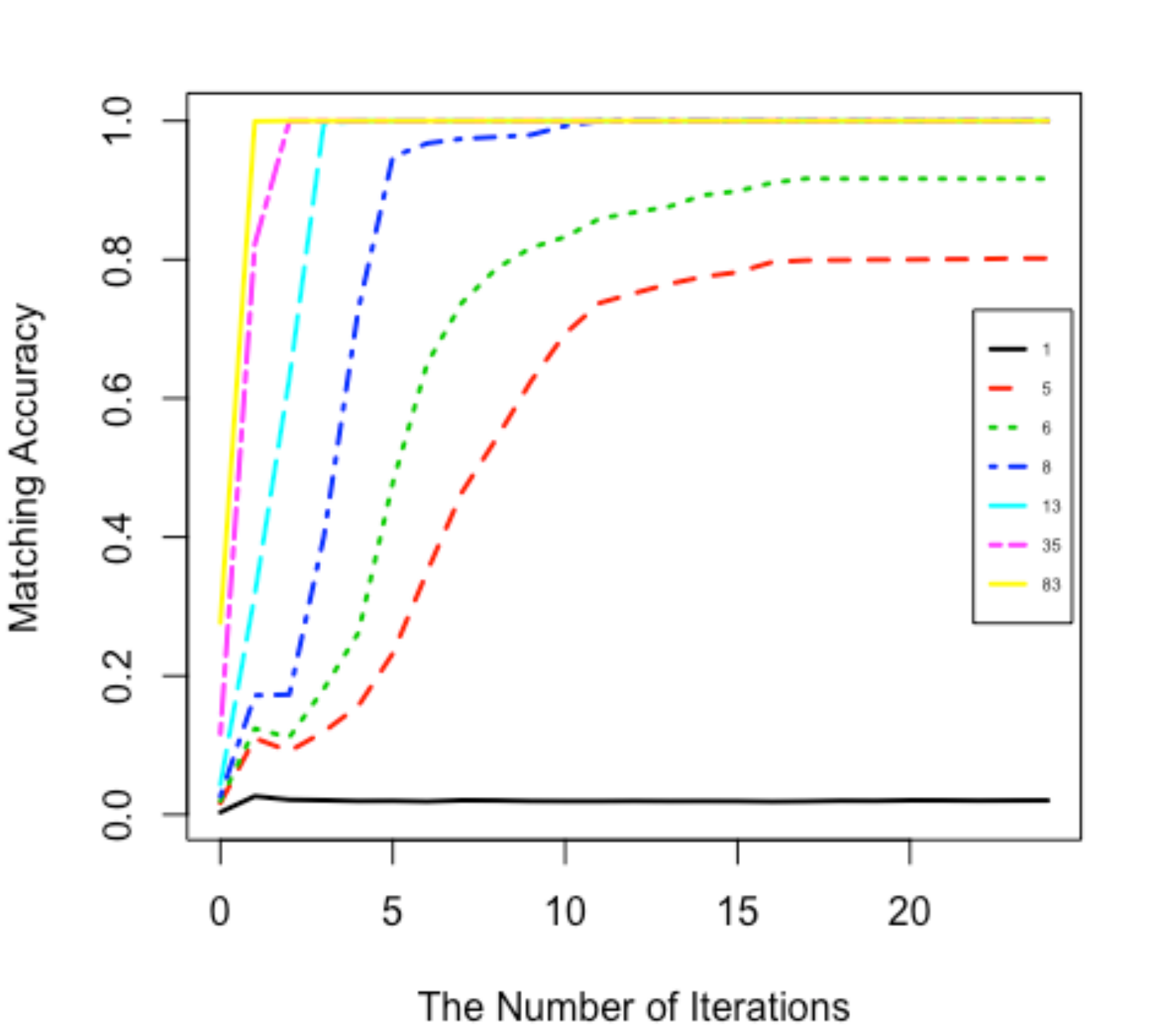}
    \caption{ RDPG}
    \label{fig:rdpg}
  \end{subfigure}\hfill
  \begin{subfigure}[t]{.48\textwidth}
    \centering
    \includegraphics[width=\linewidth]{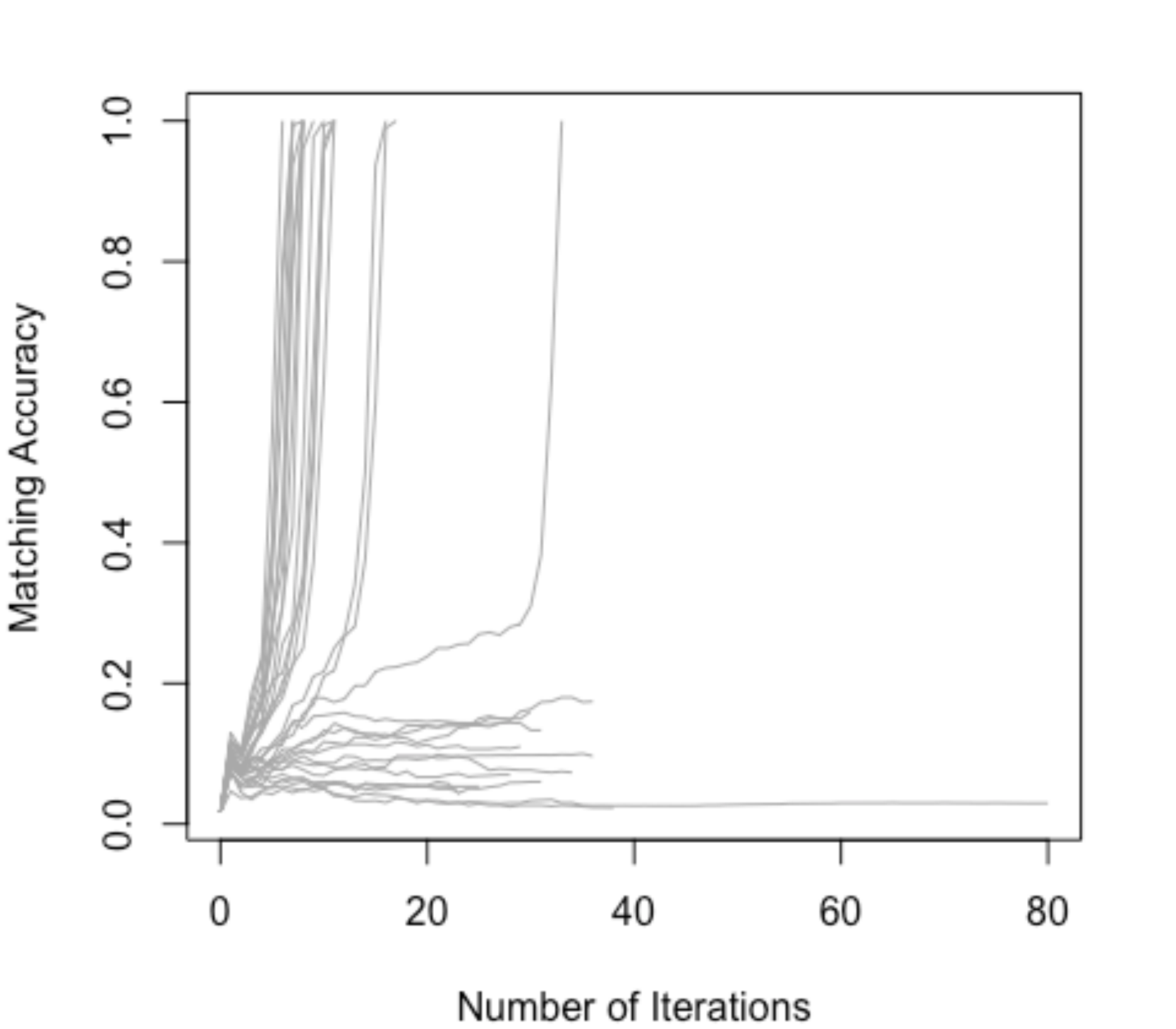}
    \caption{RDPG accuracy trajectories}
    \label{fig:rdpg_traj}
  \end{subfigure}
  \caption{The left column shows the average matching accuracy across iterations for the three settings with $\tr(D^0)\in\{1,5,6,8,13,35,83\}$ for 30 Monte Carlo replicates. The right columns shows how the accuracy of each of these thirty replicates behaves for $s=13$ for the HER setting and $s=5$ for the SBM and RDPG settings.}
\label{fig:traj}
\end{figure}

In the left panels of Figure~\ref{fig:traj}, we plot the matching accuracy (averaged across $m=30$ Monte Carlo replicates) across each iteration of the FAQ algorithm as $s$ varies for the three random graph settings.
The accuracy is plotted from iterations 0 to 20, with the accuracy at iteration 0 simply given by $s/n$.
Figure~\ref{fig:er} top, Figure~\ref{fig:sbm} middle, and Figure~\ref{fig:rdpg} bottom, shows the average accuracy for the HOM, SBM, and RDPG settings respectively.

In all settings, the performance clearly increases as $s$ increases.
For the homogeneous \ErdosRenyi{} case, there is a clear distinction between the performance at 8, 13, and 35 seeds.
For 8 and fewer seeds, FAQ fails to perform well, maintaining accuracy nearly equal to the initial accuracy.
In these cases, FAQ is terminating at a local maxima far from the identity.
For 35 and greater seeds, FAQ almost immediately converges to perfect accuracy as suggested by the theory.
At 13 seeds we see intermediate performance, taking more iterations until an intermediate doubly stochastic matrix with sufficiently large trace is reached, and good performance is achieved.
For the SBM and RDPG cases, we see somewhat similar behavior but even with only 5 seeds, FAQ achieves excellent but not perfect performance after 10 iterations.
Overall, the SBM setting is somewhat easier then the RDPG setting which is somewhat easier then the HOM setting.
This is likely due to the case that for these settings $\tr(\Lambda D^0 \Lambda P)$ will typically increase with $\tr(P)$ so that the structure $\Lambda$ is actually making the matching easier.

In the SBM and RDPG settings, the average accuracy with a small number of seeds settles between $0.6$ and $1$.
The right panels of Figure~\ref{fig:traj} shows the accuracy at each iteration for the 30 different Monte Carlo replicates. 
FAQ starts with $s=13$ for the HOM setting and $s=5$ for the SBM and RDPG settings.
In the HOM setting, all 30 replicates eventually achieve perfect performance, with most stopping after 4--7 iterations.
However, in the SBM case 2 of the 30 replicates end with less than $0.18$ accuracy and in the RDPG case, 13 of the 30 replicates end with less than $0.2$ accuracy (see Table~\ref{tab:traj_acc}).
This matches very well with out theoretical results.
In particular, we see that under some scenarios it may take a while to get above the threshold accuracy (given by the trace constraints in Theorems \ref{thm:hom} and \ref{thm:het}), but once it does, the algorithm will rapidly converge to perfect performance in one or two steps. 
Additionally, in the RDPG and SBM settings, the replicates that terminate at a local maximum all have markedly lower accuracy; this confirms the findings discussed in Section~\ref{sec:loc_max}.

\begin{table} \centering 
  \caption{Summary of performance for trajectories illustrated in Fig \ref{fig:sbm_traj} and \ref{fig:rdpg_traj}.} 
  \label{tab:traj_acc}
  \begin{tabular}{rll}
    \toprule
    & SBM with seeds 5& RDPG with seeds 5\\
    \midrule
perfect trajectories & 28 & 17  \\
other trajectories & 2 with $ \leq 15 \%$ & 13 with $ \leq 18 \%$ \\ 
    
    \bottomrule
  \end{tabular} 
\end{table} 

In addition to investigating the performance for a fixed number of nodes, we also wanted to investigate performance as the number of nodes in the graph grows.
In Figure~\ref{fig:er_n}, we plot the accuracy in the HOM setting as a function of the number of seeds and the number of nodes.
The initial matrix $D^0$ is constructed as described above with $s\in \{2, 3, \dotsc, 50\}$.
The total number of nodes in the network varies from $n=100$ to $700$, and $A,B\sim\CER(0.5 J_n, 0.5 J_n)$.
For each $n$, we sample $500$ graphs and computed the accuracy using the varying number of seeds.
From left to right, the color of each pixel indicates the average accuracies after iteration 1, iteration 2 and after the final iteration respectively.
Note, the scale of the y-axis (i.e., the number of seeds) changes in each panel.

While our theoretical results do not directly imply that below a certain trace threshold the FAQ algorithm will not perform well, Figure~\ref{fig:er_n} provides empirical evidence for the validity of this assertion.
Performance after the first iteration varies relatively smoothly with the number of seeds, but performance after the second iteration, and especially after the final iteration, has a somewhat sharp phase transition.
Below the phase transition, performance is effectively chance, while above the phase transition the true correspondence is perfectly recovered. 
The number of seeds required before perfect performance is achieved also grows slowly, with the transition rate possibly growing slower than the $\sqrt{n}$-rate for which our theory applies.

\begin{figure}
    \begin{subfigure}[t]{.3\textwidth}
    \centering
    \includegraphics[height=.18\textheight]{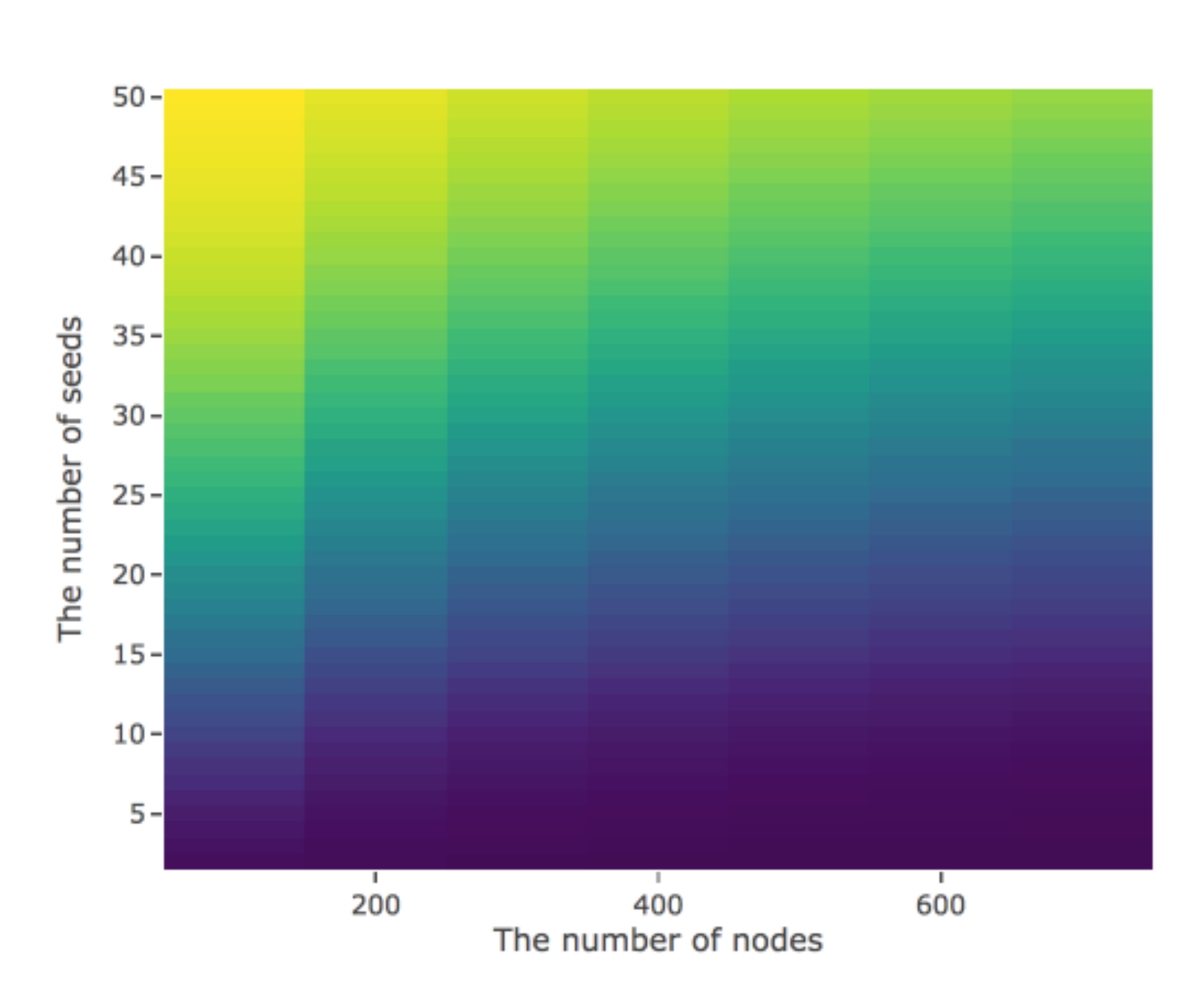}
    \caption{first iteration}
    \label{fig:er_n_first}
  \end{subfigure}
  \begin{subfigure}[t]{.3\textwidth}
    \centering
    \includegraphics[height=.18\textheight]{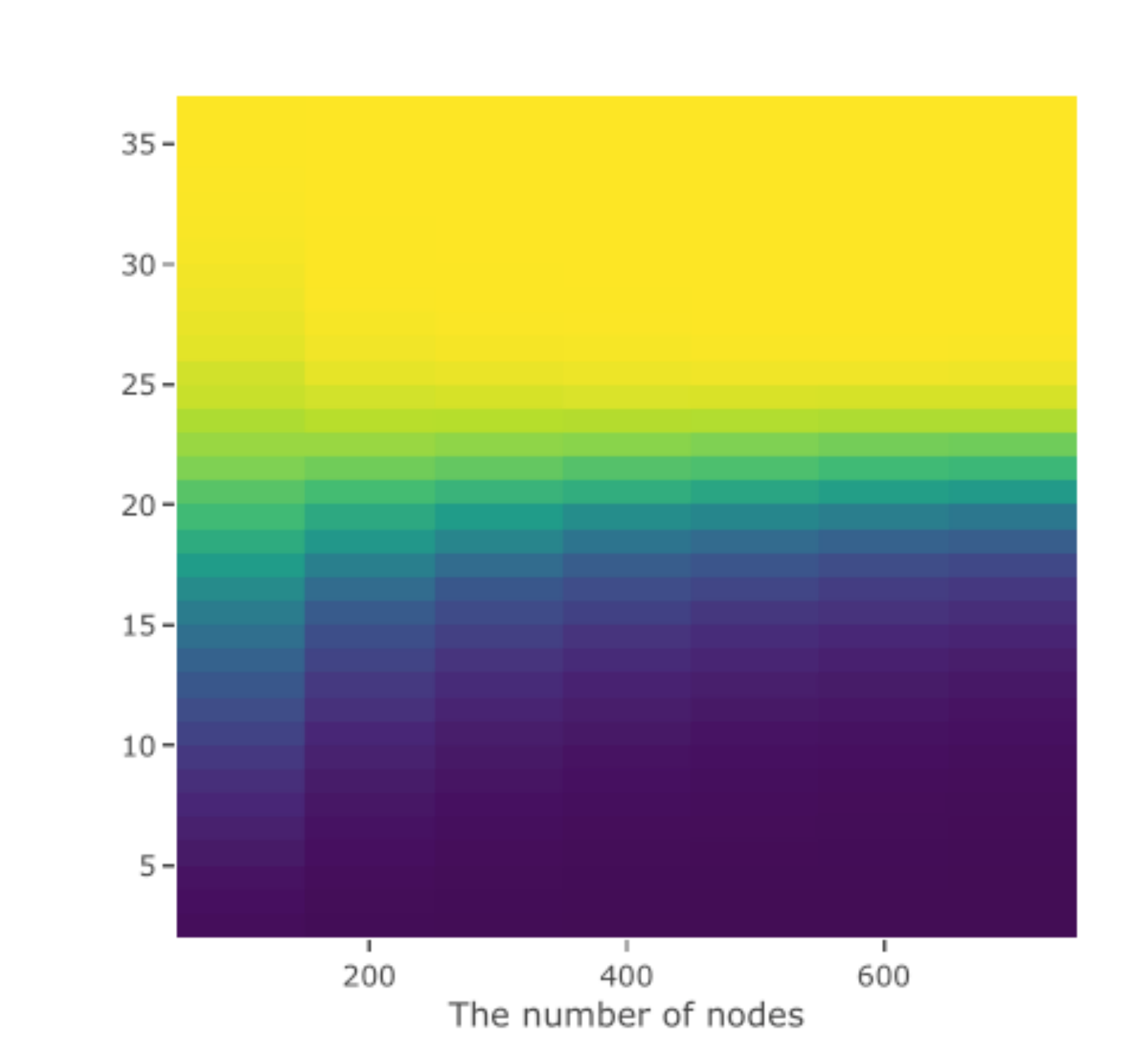}
    \caption{second iteration}
    \label{fig:er_n_sec}
  \end{subfigure}
  \begin{subfigure}[t]{.33\textwidth}
    \centering
    \includegraphics[height=.18\textheight]{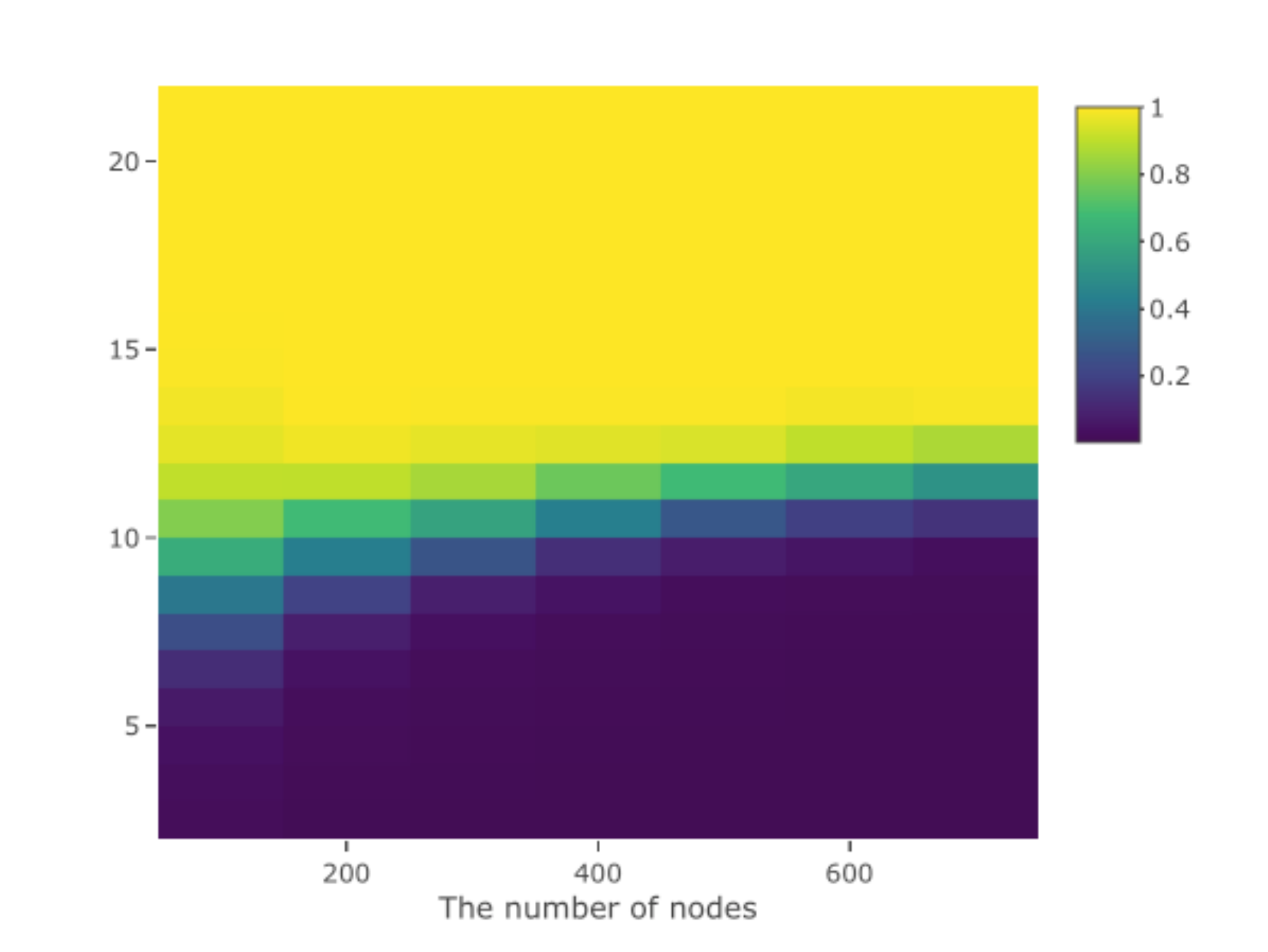}
    \caption{last iteration}
    \label{fig:er_n_last}
  \end{subfigure}
    \caption{Matching accuracy as a function of the number of seeds and the number of nodes for the homogeneous \ErdosRenyi{} setting with $A,B\sim \CER(0.5J_n, 0.5 J_n)$.
    The color of each pixel denotes the average accuracy across 500 replicates after the first, second, and last iteration, from left to right, respectively.}
    \label{fig:er_n}
\end{figure}

\subsection{Block alignment for $D^0$ and SBM}

As a final investigation, in the SBM model we will study the impact of differences between the partition of the vertex sets provided by $D^0$ and the partition of the vertices by SBM block membership.
As we noted above, matching in the SBM setting appears to be computationally easier than in the HOM setting, partially due to the fact that there is information about the true correspondence in the structure of $\Lambda$.
% The structure of $D^0$ may be similar to the structure of $\Lambda$ or not.
In the $K$-block SBM model, we consider the case that $D^0$ is formed by a $K$-part partition of the vertices that is consistent across graphs:
Specifically, for a partition $[n] = \eta_1 \sqcup \eta_2 \dotsb \sqcup \eta_K$, we set $D^0_{ij} = 1/|\eta_k|$ if $i,j$ are both in $\eta_k$ for some $k$ and $0$ otherwise.
This makes $\tr(D^0)=K$.
Suppose the blocks of the SBM are given by a different partition $\beta_1\sqcup  \beta_2 \sqcup \dotsb \sqcup \beta_K$.

We can characterize the disagreement between a $D^0$ constructed in this way and the block memberships of the vertices according to the SBM by computing the confusion matrix.
The confusion matrix for the two partitions is $\mathcal{C}\in \mathbb{N}^{K \times K}$ with $\mathcal{C}_{ij} = |\eta_i \cap \beta_j|$.
We define the disagreement between the two partitions as $d(\eta,\beta) = n-\argmax_{P\in \PM_k} \tr(\mathcal{C}P)$, the number of nodes in ``different'' partitions after relabeling the partitions to align as well as possible.
To study the impact of disagreement between $D^0$ and the SBM, we sampled partitions $\eta_1,\dotsc,\eta_m$ of $[n]$ which had confusion matrices fixed at $\frac{n-\delta}{K} I_K  + \frac{\delta}{K(K-1)}(J_K-I_K)$ for various choices of $\delta$.
Given $\delta\leq n-n/K$, the disagreement will be exactly $\delta/k$ for each part.

\begin{figure}
  % \centering
  % \begin{subfigure}[t]{.48\textwidth}
  %   \centering
  %   \includegraphics[width=\linewidth,height=6 cm]{../figure/Fig5_Block_alignment.pdf}
  %   \caption{Matching accuracy for 10 blocks and 10 seeds of SBM. The disagreements between SBM and $D^0$ within each blocks are $(0,1,2,3)$, i.e., with total disagreements $(0,90=1\times (10-1) \times 10,180,270)$. Each line in the graph represents the matching accuracy along iterations with different disagreements.  }
  %   \label{fig:iter_dis}
  % \end{subfigure}\hfill 
  % \begin{subfigure}[t]{.48\textwidth}
    \centering
    \includegraphics[width=.6\linewidth]{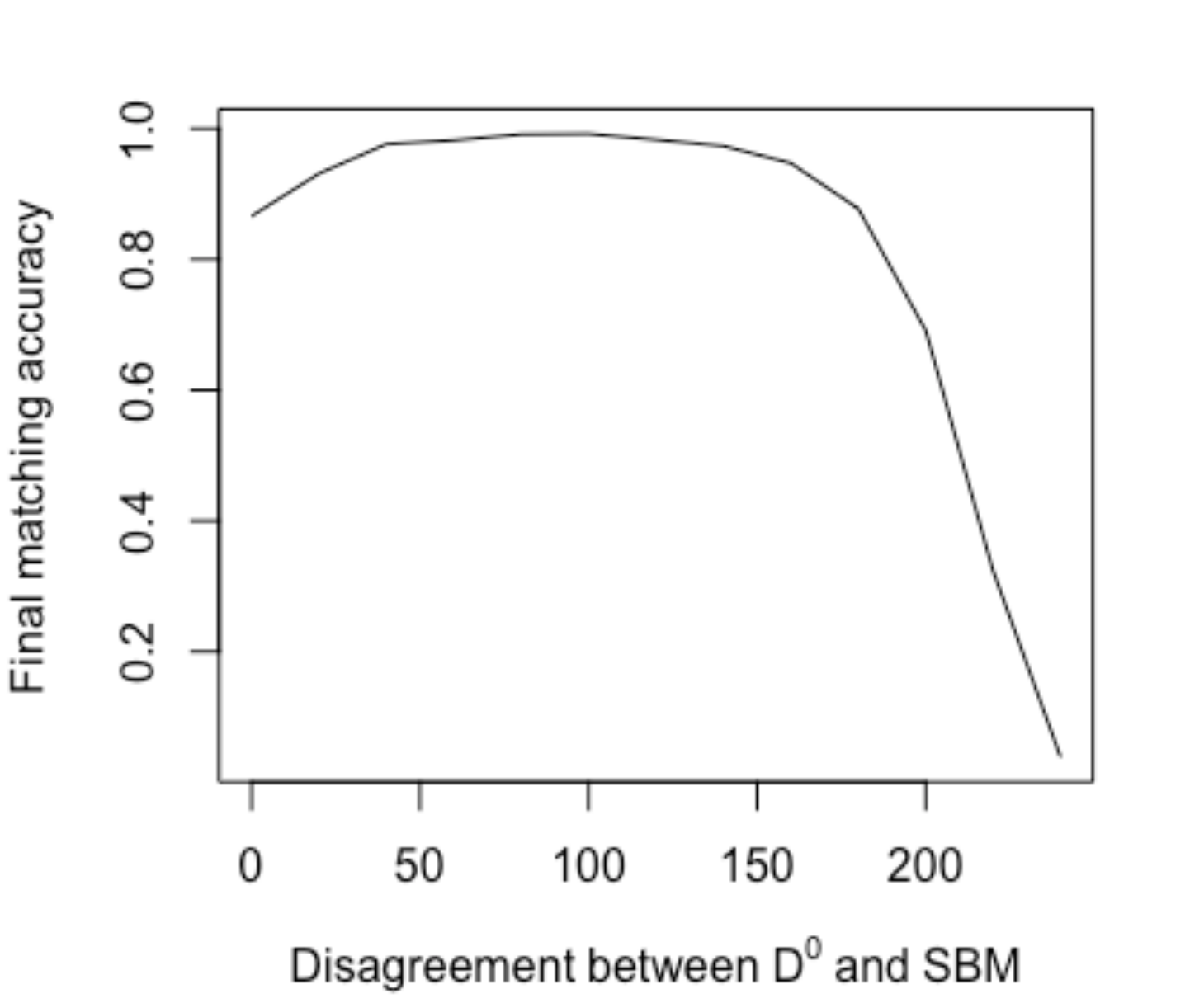}
    \caption{Final matching accuracy for 5 soft seeds, 5 blocks of SBM. The disagreements between SBM and $D^0$ within each blocks are $(0,4,8,\dotsc, 48)$, i.e., with total disagreements  $(0,20,40,60,80,100,120,140,160,180,200,$ $220,240)$.}
    \label{fig:acc_diss}
  % \end{subfigure}\hfill 
  % \break
\end{figure}

We explored the impact of this disagreement for the case of $$A,B\sim \CER \left( \left(0.5 I_K + 0.1 (J_K-I_K)\right) \otimes J_{n/K}, 0.5 J_n \right).$$
For $K=5$, Figure~\ref{fig:acc_diss} shows the average matching accuracy, after the final iteration of FAQ, as a function of the number of disagreements, with $\delta\in \{0,20, \dotsc, 240\}$.
The average is based on 100 Monte Carlo replicates.
The number of disagreements within each block of $D^0$ ranged as $0,4,8,\dotsc, 48$, with $0,1,2,\dotsc, 12$ vertices being mapped to each disagreeing block in the SBM.
Between $\delta=0$ to $60$, the matching accuracy improves slightly as more disagreements are introduced between the partitions, before leveling off.
However, for $\delta>150$, the accuracy starts to drop and drops rapidly for $\delta>200$.

We postulate that as the number of disagreements increase initially, the vertices which disagree become easily matchable.
For disagreeing blocks, there are few nodes within each intersection of $D^0$ and SBM blocks.
These nodes can rapidly add to the trace in the first few iterations, making the problem easier.
When there are many disagreements, we can view $D^0$ as effectively reducing the matching problem to matching 5 smaller graphs with nearly the same structure for $\Lambda$. 
From this perspective, the four other $D^0$ do not provide enough information to help align within each block.

Table~\ref{tab:dis} shows a more detailed breakdown of the errors that occur for $\delta=0,60,200$, and $220$.
The second column breaks down the accuracy according to the percentage of replicates where perfect performance was achieved, along with the conditional accuracy given that perfect performance was not achieved.
The rate of perfect matchings at $\delta=60$ (93.5\%) is substantially better than when $\delta=0$ (84\%) but when the matching fails it tends to produce slightly more errors.
These results reinforce those from Figure~\ref{fig:sbm_traj}.

The remaining columns show, conditioned on an imperfect alignment, whether the errors occur within or between blocks for both $D^0$ and the SBM.
An error is within a block if the vertex gets matched to another vertex in the same block and otherwise it is between blocks.
As the number of disagreements increase, the errors for $D^0$ tends to be more between blocks while the errors for the SBM are all within block, until $\delta=220$.
This indicates that for $\delta\leq 200$, even if the alignment is incorrect the alignment of vertices to the correct SBM block remains perfect.
At $\delta=220$, accuracy is even lower.
Furthermore, even though the majority of errors for SBM remain within the blocks, errors between SBM blocks begin to appear.
Note the errors for $D^0$ between blocks are large also. 
Note that for blocks of this size, a random alignment would produce on average $299$ errors of which $\approx 59$ would be within block and $\approx 240$ would be between blocks.
Hence, for $\delta=200$ and $220$, the number of disagreements for $D^0$ is close to chance.

\begin{table} \centering 
    \caption{Average matching and mismatching accuracy over mismatched replicates. The first column indicates the number of disagreements and the second column shows both the percentage of perfect and imperfect alignments, along with the conditional mean of the accuracy given the alignment was imperfect.
    The remaining columns show how errors distribute within versus between blocks for $D^0$ and the SBM.
    % {\color{blue} These numbers don't seem to quite add up. Eg. 258+43 = 301? Also, 228+131=359? The reason for these disagreements coming from when update the correct alignment, I use a different dataset. I now have checked and kept to use the same datasets. For the case d=220,   I calculate the SBM between blocks by just dividing the number of non-zero elements of SBM between blocks, which is much smaller than the number of instances of imperfect alignment (The non-zero instances of SBM between blocks are 41, but the imperfect matching instance are 146. So I should divide it by 146 instead of 41.), I have corrected it in the table.}
    } 
    \label{tab:dis}
    \begin{tabular}{@{\extracolsep{5pt}} r|c|cc|cc} 
        \toprule
        $\delta$ & correct & \multicolumn{4}{c}{errors}\\
         & alignment & \multicolumn{2}{c|}{$D^0$}& \multicolumn{2}{c}{SBM}\\
        & & within & between  & within & between \\\midrule
        % $(10,270)$ & \thead{$11 \% \ 35.18$ \\ $89 \% \ 300$}  & $34.5$ & $230$ & $62$ & $202$ \\
        $0$ & \thead{$16 \% \ 37.5$ vertices \\ $84 \%$ perfect}  & $262$ & $0$ & $262$ & $0$ \\  
        $60$ & \thead{$6.5 \% \ 41.6$  vertices \\ $93.5 \%$ perfect} & $169$ & $89$ & $258$ & $0$ \\ 
        $200$ & \thead{$42.5 \% \ 32.7$ vertices \\ $57.5 \%$ perfect } & $61$ & $206$ & $267$ & $0$ \\
        $220$ & \thead{$73 \% \ 22.3$ vertices \\ $27 \%$ perfect } & $59$ & $219$ & $242$ & $35$ \\ \bottomrule
    \end{tabular}
    %     \begin{tabular}{@{\extracolsep{5pt}} r|c|c|c|c|c|c|c} 
    %     \\[-1.8ex]\hline 
    %     \hline \\[-1.8ex]
    %     seeds, & \multirow{2}{1.6cm}{accuracy\\} & \multicolumn{6}{c}{errors}\\
    %     disagree & & \multicolumn{2}{c|}{$D^0$}& \multicolumn{2}{c|}{SBM}& \multicolumn{2}{c}{Chance} \\
    %     & & within & between  & within & between & within & between \\
    %     \hline
    %     $(10,270)$ & \thead{$11 \% \ 35.18$ \\ $89 \% \ 300$}  & $34.5$ & $230$ & $62$ & $202$ & $28$ & $270$\\
    %     \hline
    %     $(5,0)$ & \thead{$16 \% \ 37.46$ \\ $84 \% \ 300$}  & $262$ & $0$ & $262$ & $0$ & $58$ & $240$\\ 
    %     \hline 
    %     $(5,60)$ & \thead{$2 \% \ 43$ \\ $98 \% \ 300$} & $169$ & $89$ & $258$ & $0$ & $58$ & $240$\\ 
    %     \hline
    %     $(5,200)$ & \thead{$42.5 \% \ 32.705$\\ $57.5 \% \ 300$ } & $60$ & $206$ & $267$ & $0$ & $59$ & $239$\\
    %     \hline
    %     $(5,220)$ & \thead{$73 \% \ 64$\\ $27 \% \ 300$ } & $58$ & $219$ & $228$ & $131$ & $58$ & $240$\\
    %     \hline \\[-1.8ex] 
    % \end{tabular}
\end{table}

\section{Discussion}

By initializing the FAQ algorithm at a doubly stochastic matrix based on possibly noisy prior information, the algorithm can be shown to perform well provided there is enough information.
Indeed, we show that under a broad set of regimes, with high probability the algorithm will converge in two steps to the true correspondence under the correlated \ErdosRenyi{} model.
The quantity of sufficient prior information is directly related to the trace of the initial matrix as well as how that initial matrix relates to the structure of the model parameters.
Our simulations illustrate various aspects of the theory and explore how the structure of the model and the initial matrix impact performance.

One important implication of our work is that if the assumptions of our theorems hold, there will be no local optima near the true correspondence.
Furthermore, it is likely that the non-global local optima will have a graph matching objective function Eq.~\eqref{eq:gm_intro}, which is much larger than that of the global optimum.
Exploiting this fact, it may be possible to construct a test to determine whether the global optimum has been reached by using random restarts.

In Section~\ref{sec:construct_ds}, we propose the use a similarity matrix comparing nodes to construct the initial matrix.
Note that such a similarity matrix $S$ could also be directly incorporated into the objective function by modifying Eq.~\eqref{eq:gm_relax} as
\begin{equation}
  \argmax_{D \in \DSM_{n}} \tr(ADBD^T) + \tr(SD)
\end{equation}
which has gradient $ADB+S$.
Note that the hard seeding objective function, Eq.~\eqref{eq:sgm} is precisely of this form, with $S=B_{12}^T A_{12}$.
Whether $S$ is random or fixed, the proofs of our theorems can likely be modified by incorporating the relative magnitudes of the diagonal of $S$ as compared to its off-diagonal.

We do not consider non-identically distributed graphs or graphs of different sizes.
These pathologies do present substantial additional challenges but we refer the reader to other work on adapting FAQ to these setting using centering~\cite{lyzinski2017graph} and padding~\cite{Sussman2018-lb}.
These techniques are easily combined with soft seeding.
A more challenging future direction is to understand what prior information should and should not be used in constructing the matrix $D^0$.
This may be viewed as trying to learn the kernel to construct the similarity matrix $S$.
Using known hard seeds, we may be able to determine which features of the covariates are most useful for determining the true correspondence.
Techniques from entity resolution \cite{Getoor2012-es} can be exploited for this task.
A deeper understanding of how the relationship between $D^0$, $\Lambda$ and $R$ impacts performance is also currently elusive but would likely aid in determining what resources to use to construct $D^0$ when estimating $\Lambda$ may be possible.

\section{Acknowledgments}
This work was partially supported by a grant from MIT
Lincoln Laboratories and the Department of Defense, and the MAA
program of the Defense Advanced Research Projects Agency (DARPA) via award number FA8750-18-0035.
The views and conclusions contained herein are those of the authors and should not be interpreted as necessarily representing the official policies or endorsements, either expressed or implied, of the Air Force Research Laboratory and DARPA, or the U.S.Government.

 \bibliographystyle{plain}
\bibliography{../biblio.bib}

\begin{thebibliography}{10}

\bibitem{alex}
Y.~Aflalo, A.~Bronstein, and R.~Kimmel.
\newblock On convex relaxation of graph isomorphism.
\newblock {\em Proceedings of the National Academy of Sciences},
  112(10):2942--2947, 2015.

\bibitem{ConteReview}
D.~Conte, P.~Foggia, C.~Sansone, and M.~Vento.
\newblock Thirty years of graph matching in pattern recognition.
\newblock {\em International Journal of Pattern Recognition and Artificial
  Intelligence}, 18(03):265--298, 2004.

\bibitem{Durante2018-pl}
D.~Durante and D.~B. Dunson.
\newblock Bayesian inference and testing of group differences in brain
  networks.
\newblock {\em Bayesian analysis}, 13(1):29--58, March 2018.

\bibitem{Egozi2013-jh}
A.~Egozi, Y.~Keller, and H.~Guterman.
\newblock A probabilistic approach to spectral graph matching.
\newblock {\em IEEE transactions on pattern analysis and machine intelligence},
  35(1):18--27, January 2013.

\bibitem{Emmert-Streib2016-st}
F.~Emmert-Streib, M.~Dehmer, and Y.~Shi.
\newblock Fifty years of graph matching, network alignment and network
  comparison.
\newblock {\em Information sciences}, 346--347:180--197, 2016.

\bibitem{feizi2016spectral}
S.~Feizi, G.~Quon, M.~Recamonde-Mendoza, M.~M{\'e}dard, M.~Kellis, and
  A.~Jadbabaie.
\newblock Spectral alignment of networks.
\newblock {\em arXiv preprint arXiv:1602.04181}, 2016.

\bibitem{finke1987quadratic}
G.~Finke, R.~E. Burkard, and F.~Rendl.
\newblock Quadratic assignment problems.
\newblock {\em North-Holland Mathematics Studies}, 132:61--82, 1987.

\bibitem{FAP}
D.E. Fishkind, S.~Adali, H.~G. Patsolic, L.~Meng, V.~Lyzinski, and C.E. Priebe.
\newblock Seeded graph matching.
\newblock {\em arXiv preprint arXiv:1209.0367}, 2017.

\bibitem{FW}
M.~Frank and P.~Wolfe.
\newblock An algorithm for quadratic programming.
\newblock {\em Naval Research Logistics Quarterly}, 3(1-2):95--110, 1956.

\bibitem{Getoor2012-es}
L.~Getoor and A.~Machanavajjhala.
\newblock Entity resolution: Theory, practice \& open challenges.
\newblock {\em Proceedings of the VLDB Endowment International Conference on
  Very Large Data Bases}, 5(12):2018--2019, August 2012.

\bibitem{Jonker1987-bs}
R.~Jonker and A.~Volgenant.
\newblock A shortest augmenting path algorithm for dense and sparse linear
  assignment problems.
\newblock {\em Computing}, 38(4):325--340, December 1987.

\bibitem{kazemi2015growing}
E.~Kazemi, S.~H. Hassani, and M.~Grossglauser.
\newblock Growing a graph matching from a handful of seeds.
\newblock {\em Proceedings of the VLDB Endowment}, 8(10):1010--1021, 2015.

\bibitem{kim}
J.~H. Kim, B.~Sudakov, and V.~H. Vu.
\newblock On the asymmetry of random regular graphs and random graphs.
\newblock {\em Random Structures and Algorithms}, 21:216--224, 2002.

\bibitem{Kivela2014-ij}
M.~Kivel{\"a}, A.~Arenas, M.~Barthelemy, J.~P. Gleeson, Y.~Moreno, and M.~A.
  Porter.
\newblock Multilayer networks.
\newblock {\em Journal of Complex Networks}, 2(3):203--271, September 2014.

\bibitem{hungarian}
H.~W. Kuhn.
\newblock The {H}ungarian method for the assignment problem.
\newblock {\em Naval Research Logistic Quarterly}, 2:83--97, 1955.

\bibitem{Levin2017-uk}
K.~Levin, A.~Athreya, M.~Tang, V.~Lyzinski, and C.~E. Priebe.
\newblock A central limit theorem for an omnibus embedding of random dot
  product graphs.
\newblock {\em arXiv preprint, arXiv:1705.09355}, 2017.

\bibitem{lyzinski2016information}
V.~Lyzinski.
\newblock Information recovery in shuffled graphs via graph matching.
\newblock {\em IEEE Transactions on Information Theory}, 64(5):3254--3273,
  2018.

\bibitem{lyzinski2014seeded}
V.~Lyzinski, S.~Adali, J.~T. Vogelstein, Y.~Park, and C.~E. Priebe.
\newblock Seeded graph matching via joint optimization of fidelity and
  commensurability.
\newblock {\em arXiv preprint arXiv:1401.3813}, 2014.

\bibitem{Lyzinski2016-kp}
V~Lyzinski, D~E Fishkind, M~Fiori, J~T Vogelstein, C~E Priebe, and G~Sapiro.
\newblock Graph matching: Relax at your own risk.
\newblock {\em IEEE transactions on pattern analysis and machine intelligence},
  38(1):60--73, 2016.

\bibitem{lyzinski2017graph}
V.~Lyzinski and D.~L. Sussman.
\newblock Matchability of heterogeneous networks pairs.
\newblock {\em arXiv preprint, arXiv:1705.02294}, 2017.

\bibitem{Qiao2018-xy}
Zihuan Qiao, Daniel~L Sussman, Vince Lyzinski, Park Youngser, and Joshua
  Agterberg.
\newblock {iGraphMatch}, May 2018.

\bibitem{shirani2017seeded}
F.~Shirani, S.~Garg, and E.~Erkip.
\newblock Seeded graph matching: Efficient algorithms and theoretical
  guarantees.
\newblock In {\em Signals, Systems, and Computers, 2017 51st Asilomar
  Conference on}, pages 253--257. IEEE, 2017.

\bibitem{Sinkhorn1967-eo}
R.~Sinkhorn and P.~Knopp.
\newblock Concerning nonnegative matrices and doubly stochastic matrices.
\newblock {\em Pacific Journal of Mathematics}, 21(2):343--348, 1967.

\bibitem{Sussman2018-lb}
D.~L. Sussman, V.~Lyzinski, Y.~Park, and C.~E. Priebe.
\newblock Matched filters for noisy induced subgraph detection.
\newblock {\em arXiv preprint, arXiv:1803.02423}, 2018.

\bibitem{FAQ}
J.~T. {Vogelstein}, J.~M. {Conroy}, V.~{Lyzinski}, L.~J. {Podrazik}, S.~G.
  {Kratzer}, E.~T. {Harley}, D.~E. {Fishkind}, R.~J. {Vogelstein}, and C.~E.
  {Priebe}.
\newblock {Fast Approximate Quadratic Programming for Graph Matching}.
\newblock {\em PLoS ONE}, 10(04), 2014.

\bibitem{Vogelstein2015-xh}
J.~T. Vogelstein and C.~E. Priebe.
\newblock Shuffled graph classification: Theory and connectome applications.
\newblock {\em Journal of Classification}, 32(1):3--20, 2015.

\bibitem{Wang2010-og}
F.~Wang, P.~Li, and A.~C. Konig.
\newblock Learning a {Bi-Stochastic} data similarity matrix.
\newblock In {\em 2010 {IEEE} International Conference on Data Mining}, pages
  551--560. ieeexplore.ieee.org, December 2010.

\bibitem{yartseva2013performance}
L.~Yartseva and M.~Grossglauser.
\newblock On the performance of percolation graph matching.
\newblock In {\em Proceedings of the first ACM conference on Online social
  networks}, pages 119--130. ACM, 2013.

\end{thebibliography}

\appendix

\section{Propositions}
We first supplement two important propositions which will be used in the following sections.

The first is a strong concentration inequality for functions of independent Bernoulli random variables.
\begin{proposition}[\cite{kim}]
  \label{prop:kim}
  Let $X_1, X_2,..., X_m$ be a sequence of independent Bernoulli random variables with parameters $p_1, p_2, ..., p_m$. Let $f: \{0,1 \}^m \rightarrow \mathbb{R}$ be a function of $X_1, X_2,..., X_m$ and by changing any $X_i$ to $1-X_i$, f at most change
  \[ M=\sup_{i\in [m]} \sup_{(X_1, X_2,..., X_m)}  |f(X_1,..., X_{i}, ...., X_{m})-  f(X_1,..., 1-X_{i}, ...., X_{m}) |  \]
  Let $\gamma^2=M^2\sum_{i=1}^m p_{i}(1-p_i)$. 
  Then $\Pr\left[ \ |f-\mathbb{E}(f)|\geq t\gamma \ \right] \leq 2e^{-\frac{t^2}{4}}$ for all $0<t<\frac{2\gamma}{M}$.
\end{proposition}

To employ Proposition~\ref{prop:kim}, we supplement with this result which states that a pair of dependent Bernoulli random variables can always be written as a function of three independent Bernoulli random variables.
We also provide a relatively loose but convenient upper bound on the variance of these independent Bernoullis.
\begin{proposition}
  \label{prop:bibern}
  If $X,Y$ are marginally Bernoulli random variables with parameter $\Lambda$ and correlation $\rho$, then the bivariate Bernoulli distribution of $(X,Y)$ is as following
  
    \begin{center}
      \begin{tabular}{c|cc|c}
        \toprule
        $(X,Y) $ & $X=1$ & $X=0$& Total\\
        \midrule
        $Y=1$ & $\Lambda[\Lambda+\rho(1-\Lambda)]$ &  $(1-\Lambda)\Lambda(1-\rho)$ &  $\Lambda$\\
        $Y=0$ & $\Lambda(1-\rho)(1-\Lambda)$ & $(1-\Lambda+\Lambda \rho)(1-\Lambda)$& $1-\Lambda$\\
        \midrule
        Total & $\Lambda$ & $1-\Lambda$ & 1\\\bottomrule
      \end{tabular}
    \end{center}
  
  Let $Z_0, Z_1$ and $Z_2$ be independent Bernoulli random variables with $Z_0 \sim \mathrm{Bern}(\Lambda)$, $Z_1\sim \mathrm{Bern} \left(\Lambda(1-\rho)\right)$ and $Z_2\sim \mathrm{Bern}\left(\mathrm{\Lambda+\rho(1-\Lambda)}\right)$. Then $$(X,Y) \overset{d}{\sim} (Z_0, (1-Z_0)Z_1+Z_0Z_2)$$ Moreover, an upper bound for the variance  of the sum of $Z_0, Z_1$ and $Z_2$ is
  \begin{flalign*}
  \mathrm{Var}(Z_{0})+\mathrm{Var}(Z_{1})+ \mathrm{Var}(Z_{2}) & %\\
  % &=\Lambda(1-\Lambda)+\Lambda(1-\rho)(1-\Lambda(1-\rho))+\left[\Lambda+\rho(1-\Lambda)     \right]\left[1-\Lambda-\rho(1-\Lambda) \right]\\
  % &\leq \Lambda(1-\Lambda)+\Lambda(1-\Lambda+\Lambda\rho)+(\Lambda+\rho)(1-\Lambda)\\
 \leq 3\Lambda(1-\Lambda)+2\rho.
  \end{flalign*}    
  % \begin{flalign*}
  % & \mathrm{var}(Z_{0})+\mathrm{var}(Z_{1})+ \mathrm{var}(Z_{2}) &\\
  %  & \geq (1-\rho)\Lambda(1-\Lambda)+\Lambda(1-\rho)(1-\Lambda)+\Lambda(1-\Lambda-\rho(1-\Lambda))\\
  % & \geq 3(1-\rho)\Lambda(1-\Lambda)
  % \end{flalign*}

\end{proposition}

\section{Main Lemma}

In this section we prove a generic lemma bounding the probability that a given permutation will have larger directional derivative than another with larger trace.
This Lemma has a number of assumptions that we will later argue are reasonable in a number of situations.

\begin{lemma} \label{lem:main}
Let $D$ be a doubly stochastic matrix with $\tr(D)=m$, $P$ and $Q$ be permutation matrices with $n-k = \tr(P) \leq \tr(Q) = n - l$.
Let $d = \| P - Q \|_1$.

For $A,B\sim \CER(\Lambda, R)$, let $\epsilon = \max_{i,j} 3\Lambda_{ij}(1-\Lambda_{ij}) + 2R_{ij}$.
Suppose that $\Ex[ \tr(ADB(Q-P)) ] > C (k-l)m$ for some $0<C\leq \min_{i,j} R_{ij} \Lambda_{ij}(1-\Lambda_{ij})$.
It holds that
\begin{equation}
    \Pr[\tr(ADB(Q-P)) \leq 0] \leq 2 \exp\left\lbrace - \frac{{C^2(k-l)}^2 m^2}{48\epsilon  nd} \right\rbrace.
\end{equation}

For the case of $Q=I$, we have
\begin{equation}
    \Pr[\tr(ADB(I-P)) < 0] \leq 2 \exp\left\lbrace - \frac{C^2k m^2}{4\epsilon n} \right\rbrace.
\end{equation}
\end{lemma}

\begin{proof}
First, we will establish bounds on the variance of $$X = \tr(ADB(Q-P)) = \sum_{i=1}^{n}\sum_{j=1}^{n}\sum_{k=1}^{n} A_{ij} D_{jk} (B_{k\tau(i)} - B_{k\sigma(i)}),$$
where $\sigma,\tau:[n]\mapsto [n]$ are the permutation functions corresponding to $Q$ and $P$, respectively.

By Proposition~\ref{prop:bibern}, we have that $X$ is a function of independent Bernoulli random variables.
Indeed, for each $i,j \in [n]$.
$$(A_{ij},B_{ij})\overset{\mathcal{L}}{=}\left(Z_{0ij},(1-Z_{0ij})Z_{1ij}+Z_{0ij}Z_{2ij}\right)$$
where
\begin{align*}
  Z_{0ij} \sim \mathrm{Bern}(\Lambda_{ij}),\  &Z_{1ij}\sim \mathrm{Bern} \left(\Lambda_{ij}(1-\rho_{ij})\right) \\
  \text{and } &Z_{2ij}\sim \mathrm{Bern}\left(\mathrm{\Lambda_{ij}+\rho_{ij}(1-\Lambda_{ij})}\right).
  \end{align*}
  Note that $\mathrm{Var}(Z_{0ij}+Z_{1ij} + Z_{2ij}) \leq \epsilon$ and $\mathrm{Var}(A_{ij}) = \mathrm{Var}(B_{ij}) > C$ for all $i,j\in[n]$.

 $X$ depends on depends entries of $A$ and $B$ with indices in the set
\begin{equation}
    \left\{ \{ i, j \}, \{\tau(i),k\},\{\sigma(i), k\}: i,j, k\in[n], \tau(i)\neq \sigma(i)\right\} \label{eq:rv_index}
\end{equation}
which has cardinality upper bounded by 
$3 \left(nd-\binom{d}{2}\right)$.
In the case that $\tau(i)=i$, this sets cardinality is bounded by $n d-\binom{d}{2} = nk-\binom{k}{2}$, as $\{\{ i, j \}\text{ s.t. }\sigma(i)\neq i \}=\{ \{ \sigma(i), j \}\text{ s.t. }\sigma(i)\neq i \}$.
Let $\gamma^2$ denote the sum of the variances of these independent random variables which $X$ depends on.
To upper bound the variance for Proposition~\ref{prop:kim}, using Proposition~\ref{prop:bibern} we have 
\begin{equation*}
\gamma^2 \leq 3 M^2 \left[nd-\binom{d}{2}\right]\epsilon={\gamma'}^2,
\end{equation*}
where $M=4$ is the largest possible absolute change of $\tr(AD B (Q-P))$ by changing the value of any one of the independent Bernoulli random variables as described in Proposition~\ref{prop:bibern}.
Again, when $Q=I$ we can take $M=2$ instead.

We will also use that $\gamma > C (nd - \binom{d}{2})$, since $X$ will depend on at least $nd -\binom{d}{2}$ distinct $A_{ij}, B_{ij}$.
Applying Proposition~\ref{prop:kim} yields
\begin{align*}
      \Pr[X\leq 0] 
&\leq \Pr \left[  |X- \mathbb{E}X | \geq \frac{\mathbb{E}X}{\gamma} \gamma  \right]
\leq  \Pr \left[  |X- \mathbb{E}X | \geq \frac{C(k-l)m}{\gamma^\prime} \gamma  \right] \\
& \leq  2\exp \left[-\frac{(k-l)^2 m^2}{48\epsilon \left(nd - \binom{d}{2}\right)} C^2\right]
\end{align*}  
as long as  $(k-l)m C/\gamma^\prime \leq 2 \gamma$.
This last condition is equivalent to $m \leq \frac{2\gamma\gamma^\prime}{(k-l)C}$ which must hold since
$$m\leq n
% \leq \frac{ \epsilon_1}{C \epsilon_1} n 
\leq \frac{2C}{C} \left(n-\frac{d-1}{2} \right) 
= \frac{2\left(nd-\binom{d}{2}\right)C}{dC}
\leq \frac{2 \gamma^2}{(k-l)C} 
 \leq \frac{2\gamma \gamma'}{(k-l)C}.$$
\end{proof}

\section{Uniform Bound}

In this section the goal is to prove Theorems~\ref{thm:hom} and \ref{thm:het} which state that under certain conditions, with high probability the FAQ algorithms will converge in two steps to the identity permutations.
We will prove this over four separate stages.

The first stage shows that for any fixed $D^0$ with large trace, the FAQ algorithm will converge in one step with high probability.
The second stage shows that this will happen uniformly for all $D^0$ with trace on the order of $n-o(n)$.
The last two stages state that uniformly, for $D^0$ with smaller trace, the first step of the algorithm will be in the direction of a permutation with trace at least $n-o(n)$, and that the line search will take the largest possible step.
By taking a union bound over the events in these last three results, we arrive at a proof Theorem~\ref{thm:het}.

\subsection{In one step}

Our first result is for fixed $D^0$ with sufficiently large trace.

\begin{theorem}\label{thm:pw}
For $A,B\sim \CER(\Lambda, R)$, let $\epsilon = \max_{i,j} 3\Lambda_{ij}(1-\Lambda_{ij}) + 2R_{ij}$.
Let $D^0$ be a permutation matrix.
Suppose that for all permutation matrices $P$, $\Ex[ \tr(AD^0 B(I-P)) ] > C \tr(D^0)\tr(I-P)$ for some $0<C\leq \min_{i,j} R_{ij} \Lambda_{ij}(1-\Lambda_{ij})$.
Suppose also that $\tr(D^0)>8\sqrt{\epsilon n\log (n)}/C $.

The first step of Frank-Wolfe procedure starting at $D^0$ will yield the identity matrix with probability at least $1 - 4 \exp\left[- \frac{C^2}{4n\epsilon}{\tr(D^0)}^2 \right]$.
\end{theorem}
\begin{proof}
Let $m=\tr(D^0)$.
By Lemma~\ref{lem:main}, $\Pr[\tr(AD^0B(I-P)) \leq 0] < 2\exp \left\lbrace -\frac{C^2km^2}{4 \epsilon n}\right\rbrace$ for each $P$ with $\tr(P) = n - k$.
Note that for each $k$ the number of such $P$ is bounded by $n^k$.
Using a union bound we have 
\begin{align*}
 & \Pr[\min_{P\neq I} \tr(AD^0B(I-P)) \leq 0 ] \\
= & \Pr[ \min_{k>0} \min_{P:\tr(P)=n-k} \tr(AD^0B(I-P)) \leq 0 ] \\
\leq & \sum_{k=2}^n n^k \max_{P:\tr(P)=n-k} \Pr[ \tr(AD^0B(I-P)) \leq 0 ] \\
\leq & 2 \sum_{k=2}^n \exp\left\{k \left( \log(n)- \frac{C^2m^2}{4\epsilon n} \right)\right\}\\
\leq & 2 \sum_{k=2}^n \exp\left\{ - \frac{C^2m^2}{4\epsilon n} \right\} \leq n \exp\left\{ - \frac{C^2m^2}{4\epsilon n} \right\}.
\end{align*}
The penultimate inequality is true by the assumption on $\tr(D^0)$ being sufficiently large, and the final inequality is due to the exponent being maximized at $k=2$.

This establishes that the direction for the line search will be towards the identity matrix with high probability.
Provided this is the line search direction, the Frank-Wolfe iteration is given by maximizing the quadratic
$$\tr(A(\alpha I + (1-\alpha) D^0) B (\alpha I + (1-\alpha) D^0)^T)$$
as a function of $\alpha$ in $[0,1]$.
The derivative of the above with respect to $\alpha$ is 
\begin{align*}
2\alpha &\tr(AB - A{D^0}B - AB{D^0}^T + A{D^0}B{D^0}^T)\\ + &\tr(A{D^0}B +AB{D^0}^T - 2A{D^0}B{D^0}^T).
\end{align*}
To show that this function is increasing in $\alpha$ we will show that the derivative is positive at $\alpha=0$ and $\alpha=1$.

At $\alpha =0$ the derivative is $\tr(A{D^0}B +AB{D^0}^T - 2A{D^0}B{D^0}^T)=\tr(AD^0 B(I-D^0)^T)+\tr(A(I-D^0)B{D^0}^T) = 2\tr(AD^{0}B{(I-D^0)}^T$ which is positive by the above arguments since $D^0$ is a permutation matrix not equal to the identity.
Note this is the derivative of the objective function at the current point and, as shown above, $I$ is an ascent direction.

Similarly, when $\alpha = 1$. the derivative is $\tr(2AB - A{D^0}B - AB{D^0}^T)=2\tr(AIB(I-{D^0})^T)$.
Employing Lemma~\ref{lem:main} again, this can be shown to be non-negative with the same (or higher) probability.
% {\color{red} We can probably show that the assumptions of Lemma~\ref{lem:main} are always automatically satisfied for $D^0=I$.}
A union bound over the events that $I$ is the optimal ascent direction, and that the optimal step size takes $\alpha=1$ completes the results.

\end{proof}

The following corollary is a straightforward consequence of Theorem~\ref{thm:pw} using a union bound over all permutations with trace of order $n-o(n)$.

\begin{corollary}\label{cor:unif_os}
For $A,B\sim \CER(\Lambda, R)$, let $\epsilon = \max_{i,j} 3\Lambda_{ij}(1-\Lambda_{ij}) + 2R_{ij}$.
Let $t>0$  and $C \in (0, \min_{i,j} R_{ij} \Lambda_{ij}(1-\Lambda_{ij})$, and set $m= \frac{nC^2 t}{16\epsilon \log n}>1$.
Suppose that for all permutation matrices $D^0,P$ with $\tr(D^0)>n-m$, $\Ex[ \tr(AD^0 B(I-P)) ] > C \tr(D^0)\tr(I-P)$.

 % $1-2 \exp\left[- \frac{C^2}{2n\epsilon}{\tr(D^0)}^2 \right]$.
With probability at least $1-2 \exp\left\{ -\frac{nC^2(1-t)}{2\epsilon} \right\}$, for all permutation matrices $D^0$ such that $n-\tr D^0 < m$,  the first step of Frank-Wolfe procedure started at $D^0$ will yield the identity matrix.
\end{corollary}
\begin{proof}
Apply Theorem~\ref{thm:pw} along with a union bound to all permutation matrices $D$ with less than $n-\tr(D^0) \leq m = \frac{tnC^2}{16\epsilon\log n}$ non-fixed points.
The number of such $P$ is less than $n^{m+1} < n^{2m}$ which yields the result.
\end{proof}

\subsection{In two steps}

In this section, we prove that uniformly over all initial matrices with sufficiently large trace, the FAQ algorithm will converge to the truth in two steps.

\subsubsection[Step Direction]{Step Direction}

\begin{proposition}[Step Direction]\label{prop:step_dir}
Suppose $A,B\sim \CER{}(\Lambda, R)$. 
Let $\epsilon = \max_{i,j} 3\Lambda_{ij}(1-\Lambda_{ij}) + 2R_{ij}$.
Let $\delta \in (0,1/2)$, and $C \in (0, \min_{i,j} R_{ij} \Lambda_{ij}(1-\Lambda_{ij})$ and set 
\[
\ell = 2 \sqrt{n^{1+2\delta}}\epsilon/C^2
\text{ and }
m = C^2 n^{1-\delta}\log (n) / \epsilon.
\]

Suppose that $\Ex[\tr(APB(I-Q))] \geq C \tr(P)(n-\tr(Q))$ for all permutation matrices $P,Q$ with $\tr(P)\geq \ell$ and $\tr(Q) \leq n-m$.
Suppose also that $n^{2\delta}\epsilon/C^2 \geq 2 \log n$.
With probability at least 
\[
    1  - \exp\left\{- n^{1+\delta} \log n \right\},
\] 
for every permutation $P$ with $\tr(P) > \ell$,
the first step of Frank-Wolfe procedure will be in the direction of a permutation matrix $P^\prime$ with $\tr(P^\prime) > n - m$.
\end{proposition}

\begin{proof}
Let $\PM_k$ denote the set of permutation matrices trace $=n-k$.
Similarly, let $\PM_{\leq k}$, $\PM_{\geq k}$, \ldots, denote the set of permutations with at least, at most, \ldots, trace $k$.
To show that the step direction is towards a matrix with trace $>n-m$, we bound 
\begin{align*}
% &\Pr\left[ \bigcup_{P: \tr(P)\geq \ell} \bigcup_{Q: \tr(Q) \leq n-m} \bigcap_{P':\tr(P')} \left\{ \tr(APB(P'-Q)^T) \leq 0 \right\}  \right]  \\
&\Pr\left[ \bigcup_{P \in \PM_{\geq k}} \bigcup_{Q \in \PM_{\leq n-m}} \bigcap_{P'\in \PM_{> n-m}} \left\{ \tr(APB(P'-Q)^T) \leq 0 \right\}  \right]  \\
=&\Pr\left[\min_{T\geq \ell} \min_{P\in \PM_{n-T}} \min_{k\geq m} \min_{Q\in \PM_{k}} \max_{r > n-m} \max_{P^{\prime}\in \PM_{n-r}} \tr(APB(P^{\prime}-Q)) \leq 0\right] \\
\leq &\Pr\left[\min_{T\geq \ell} \min_{P\in \PM_{n-T}} \min_{k\geq m} \min_{Q\in \PM_{k}} \tr(APB(I-Q)) \leq 0\right] \\
\leq & \sum_{T\geq \ell} \sum_{P\in \PM{n,n-T}} \sum_{k\geq m} \sum_{Q\in \PM_{n,k}} \Pr[\tr(APB(I-Q)) \leq 0] \\
\leq & \sum_{T=\ell}^n \sum_{k=m}^n  \exp\left\{(n-T+k)\log n -\frac{kC^2  T^2}{4\epsilon n} \right\}.
 \end{align*}

The last step uses Lemma~\ref{lem:main}, the fact that $|\PM_{n,k}|\leq n^k$, and a union bound.
Since $\frac{C^2 \ell^2}{4 \epsilon n} =  n^{2\delta}\epsilon /C^2 \geq 2 \log n$, we have

\begin{align*}
     & \sum_{T=\ell}^{n} \sum_{k=m}^n  \exp\left\{(n-T)\log n + k \left(\log n-\frac{C^2T^2}{\epsilon n} \right) \right\} \\
\leq & \sum_{T=\ell}^{n/2} \sum_{k=m}^n  \exp\left\{(n-T)\log n -\frac{kC^2T^2}{8\epsilon n}\right\}\\
\leq & \exp\left\{n\log n -\frac{m C^2\ell^2}{2\epsilon n} \right\},
\end{align*}
where we use that $n-T + 2 < n$ for each $T$.
Finally,  $m C^2 \ell^2/(2\epsilon n) =2 n^{1+\delta} \log n$ which gives the desired bound.
\end{proof}

\subsubsection{Step Size}

\begin{proposition}\label{prop:step_size}
Suppose $A,B\sim \CER{}(\Lambda, R)$. 
Let $\epsilon = \max_{i,j} 3\Lambda_{ij}(1-\Lambda_{ij}) + 2R_{ij}$.
Let $\delta \in (0,1/2)$, and $C \in (0, \min_{i,j} R_{ij} \Lambda_{ij}(1-\Lambda_{ij}))$ and set 
\[
\ell = 2\sqrt{n^{1+2\delta}}\epsilon/C^2
\text{ and }
m = C^2 n^{1-\delta}\log (n) / \epsilon.
\]
Suppose that $\Ex[\tr(AQB(Q-P)] \geq C \tr(P)(n-\tr(Q))$ for all permutation matrices $P,Q$ with $\tr(P)\geq \ell$ and $\tr(Q) \leq n-m$.
Suppose also that $C^2n/\epsilon > 1800 \log n$ and $n-m>n/\sqrt{2}$.

With probability at least 
\[
    1  - n^2 \exp\left\{- \frac{C^4}{576\epsilon^2}n^{2-\delta} \log n \right\} - \exp\left\{ -n^{1+\delta} \log n\right\}.
\]
the step is maximal for all starting $P$ with $\ell \leq \tr(P) \leq n-2m$ so that the next iteration of the Frank-Wolfe algorithm will start at a permutation $Q$ with $\tr(Q) > n-m$.
\end{proposition}
\begin{proof}
By Proposition~\ref{prop:step_dir}, as a function of $Q$, $\tr(APBQ)$ is maximized at some $Q$ with $\tr(Q)> n - n^{1-\delta} C^2 \epsilon/\log n$ with probability at least $1-\exp \left\{ -n^{1+\delta} \log n  \right\}$.
The Frank-Wolfe iteration is given by maximizing
$$\tr(A(\alpha Q + (1-\alpha) D^0) B (\alpha Q + (1-\alpha) D^0)^T)$$
as a function of $\alpha$ in $[0,1]$.
We will bound the probability that, regardless of $P,Q$, as long as $Q$ is a gradient ascent direction from $P$ and $Q$ has sufficiently large trace, the step will be maximal with $\alpha=1$.

The derivative of the above with respect to $\alpha$ is 
\begin{align*}
   &2\alpha \tr(AQBQ^T - APBQ^T - AQBP^T + APBP^T) \\
 + & \tr(APBQ^T +AQBP^T - 2APBP^T)
\end{align*}
which is linear in $\alpha$.
To show that this function is increasing in $\alpha$ we need only check that the derivative is positive at $\alpha=0$ and $\alpha=1$.
% At $\alpha =0$ the derivative is equal to $2\tr(APB{(Q-P)}^T)$ which we showed in Proposition~\ref{prop:step_dir} will be positive with high probability at least {\color{red} something}.
Since $Q$ is the direction for the line search starting at $P$, the gradient for the original objective function is simply the derivative with respect to $\alpha$ at $\alpha = 0$, and so must be positive.

When $\alpha = 1$ the derivative is $$\tr(2AQBQ^T - APBQ^T - AQBP^T)=2\tr(AQB(Q-P)^T),$$
and we have
\begin{align*}
&\Pr\left[\min_{\ell \leq  T \leq n-2m} \min_{P\in \PM_{n-T}} \min_{k < m} \min_{Q\in \PM_{k}}  \tr(AQB(Q-P)^T) \leq 0\right]\\
\leq & \sum_{T=\ell}^{n-2m} \sum_{k = 1}^{m-1} n^{n-T + k} \max_{P,Q} \Pr[ \tr(AQB(Q-P)^T) \leq 0 ]
\end{align*}
By Lemma~\ref{lem:main}  and our assumptions, this can be bounded by 
\begin{align}
\label{eq:bndlem7}
\sum_{T=\ell}^{n-2m} \sum_{k = 1}^{m-1}2 n^{n-T + k}  \max_{P,Q} \exp\left\lbrace - \frac{{C^2(n-T-k)}^2 (n-k)^2}{48 \epsilon  n \|Q-P\|_F^2} \right\rbrace
\end{align}
Note that $\|Q-P\|_F^2 \leq n-T+k\leq n$, which occurs when all non-fixed points are different in $P$ and $Q$.

Hence, using that $n/\sqrt{2} < n-m < n - k$ for sufficiently large $n$,  Eq. (\ref{eq:bndlem7})
is bounded by
\begin{align}
\label{eq:bndlem72}
\leq & \sum_{T=\ell}^{n-2m} \sum_{k = 1}^{m-1}2 n^{n-T + k}   \exp\left\lbrace - \frac{{C^2(n-T-k)}^2 (n-k)^2}{48\epsilon  n (n-T+k)} \right\rbrace\notag \\
\leq & \sum_{T=\ell}^{n-2m} \sum_{k = 1}^{m-1}2    \exp\left\lbrace (n-T+k)\log n - \frac{{C^2(n-T-k)}^2 n^2}{96\epsilon  n (n-T+k)} \right\rbrace
\end{align}
Note that $n-T + k <  3 (n-T-k)$ for every summand and hence Eq. (\ref{eq:bndlem72}) is bounded by
\begin{align*}
\leq & \sum_{T=\ell}^{n-2m} \sum_{k = 1}^{m-1}2   \exp\left\lbrace (n-T - k)\left(3\log n- \frac{C^2 (n-T-k)n^2}{288\epsilon n (n-T-k)}\right) \right\rbrace \\
% \leq & \sum_{T=\ell}^{n-2m} \sum_{k = 1}^{m-1}2   \exp\left\lbrace (n-T - k)(3\log n- \frac{C^2 (n-k)^2}{96\epsilon  n}) \right\rbrace \\
\leq &  \sum_{T=\ell}^{n-2m} \sum_{k = 1}^{m-1}2   \exp\left\lbrace - m\frac{C^2 n}{576\epsilon} \right\rbrace 
\end{align*}
where we use that $n-m>n/\sqrt{2}$, $C^2n / (288\epsilon) > 6 \log n$, and $n-T-k>m$.
\end{proof}

\end{document}